\documentclass[prx,aps,superscriptaddress,floatfix,tightenlines,reprint]{revtex4-2}
\usepackage{amsthm}
\usepackage{amsmath}
\usepackage{amssymb}
\usepackage{amsfonts}
\usepackage{epsfig}
\usepackage{bm}
\usepackage{times}
\usepackage[colorlinks,linkcolor=blue,anchorcolor=blue,citecolor=blue,urlcolor=blue]{hyperref}
\usepackage{braket}
\usepackage{graphicx}
\usepackage{setspace}
\usepackage{subfigure}
\usepackage{enumerate}

\theoremstyle{definition}
\newtheorem{definition}{Definition}
\newtheorem*{definition*}{Definition}

\theoremstyle{plain}
\newtheorem{theorem}{Theorem}
\newtheorem{corollary}{Corollary}
\newtheorem{lemma}{Lemma}
\newtheorem{proposition}{Proposition}

\newtheorem*{theorem*}{Theorem}
\newtheorem*{corollary*}{Corollary}
\newtheorem*{lemma*}{Lemma}
\newtheorem*{proposition*}{Proposition}

\theoremstyle{remark}
\newtheorem{remark}{Remark}

\NewDocumentEnvironment{variant}{O{theorem} D(){} m}
{\addtocounter{#1}{-1}%
\expandafter\renewcommand\csname the#1\endcsname{\ref{#3}$'$}%
\begin{#1}[#2]}
{\end{#1}}

\NewDocumentEnvironment{appthm}{O{theorem} D(){} m}
{\addtocounter{#1}{-1}%
\expandafter\renewcommand\csname the#1\endcsname{\ref{#3}}%
\begin{#1}[#2]}
{\end{#1}}

\begin{document}

\title{Markov Gap and Bound Entanglement in Haar Random State}

\author{Tian-Ren Jin}
\affiliation{Beijing National Laboratory for Condensed Matter Physics, Institute of Physics, Chinese Academy of Sciences, Beijing 100190, China}
\affiliation{School of Physical Sciences, University of Chinese Academy of Sciences, Beijing 100049, China}

\author{Shang Liu}
\affiliation{Kavli Institute for Theoretical Physics, University of California, Santa Barbara, California 93106, USA}
\affiliation{Department of Physics, California Institute of Technology, Pasadena, California 91125, USA}
\affiliation{Beijing National Laboratory for Condensed Matter Physics, Institute of Physics, Chinese Academy of Sciences, Beijing 100190, China}

\author{Heng Fan}
\email{hfan@iphy.ac.cn}
\affiliation{Beijing National Laboratory for Condensed Matter Physics, Institute of Physics, Chinese Academy of Sciences, Beijing 100190, China}
\affiliation{School of Physical Sciences, University of Chinese Academy of Sciences, Beijing 100049, China}
\affiliation{Beijing Academy of Quantum Information Sciences, Beijing 100193, China}
\affiliation{Hefei National Laboratory, Hefei 230088, China}
\affiliation{Songshan Lake Materials Laboratory, Dongguan 523808, China}

\begin{abstract}
    Bound entanglement refers to entangled states that cannot be distilled into maximally entangled states and therefore cannot directly be used in many quantum information processing protocols. We identify a relationship between bound entanglement and the Markov gap, which is introduced within holography via the entanglement wedge cross section and is related to the fidelity of the partial Markov recovery problem. We prove that a bound entangled state must have a nonzero Markov gap. Conversely, for sufficiently large systems, a state with a weakly nonzero Markov gap typically has a bound entangled or separable marginal state, where entanglement is undistillable. Furthermore, this implies that the transition from a bound entangled to a separable state originates from the properties of states with a weakly nonzero Markov gap, which may be dual to non-perturbative effects from a holographic perspective. Our results shed light on the investigation of the Markov gap and enhance interdisciplinary applications of quantum information.
\end{abstract}
\maketitle

\section{Introduction}

    Quantum entanglement, an essential feature of quantum mechanics, has attracted considerable attention across diverse fields, ranging from condensed matter physics~\cite{PhysRevLett.96.110404,PhysRevLett.96.110405} to black hole physics~\cite{PhysRevLett.96.181602,Ryu_2006}.
    Moreover, serving as a valuable resource for quantum information processing tasks, such as quantum key distribution~\cite{PhysRevLett.67.661}, quantum dense coding~\cite{PhysRevLett.69.2881}, quantum teleportation~\cite{PhysRevLett.70.1895}, entanglement swapping~\cite{PhysRevLett.68.1251}, etc., entanglement is a central concept in quantum information theory~\cite{RevModPhys.81.865,RevModPhys.91.025001}.

    In quantum entanglement theory, several operational and mathematical criteria have been developed to detect or characterize entanglement, such as the positive partial transpose (PPT) criterion~\cite{PhysRevLett.77.1413}, and the reduction criterion~\cite{PhysRevA.59.4206}.
    PPT entangled states are bound entangled, which cannot be distilled into maximally entangled states~\cite{PhysRevLett.80.5239}.
    Thus, such states are not direct resources for quantum computation and communication tasks such as teleportation.
    Since any pure bipartite entangled state is distillable~\cite{PhysRevA.53.2046}, a bound entangled bipartite state must be a mixed state.
    Its purification exhibits tripartite entanglement.
    Furthermore, it has been shown that maximally entangled states are required to produce certain bound entangled states~\cite{PhysRevLett.86.5803}.
    This property of bound entanglement implies irreversibility of the quantum entanglement, which is a unique feature distinguishing it from other quantum resources~\cite{lami2023no}.
    Negative partial transpose (NPT) states could also be bound entangled.
    In contrast, the state that violates the reduction criterion is distillable~\cite{PhysRevA.59.4206}.

    In the context of entanglement distillation, the distillable entanglement $E_D$ and the entanglement cost $E_C$ are two key measures of entanglement~\cite{PhysRevA.54.3824,Hayden_2001,PhysRevLett.84.2014}.
    The distillable entanglement $E_D$ represents the asymptotic rate of distilling Bell pairs from a given state, and the entanglement cost $E_C$ quantifies the rate of consuming Bell pairs to prepare that state.
    Recently, a quantity called Markov gap has been proposed~\cite{hayden2021markov}, which has been suggested to detect the tripartite entanglement~\cite{akers2020entanglement,PhysRevLett.126.120501}.
    The Markov gap quantifies the difference between reflected entropy and mutual information, $h(A:B) = S_R(A:B) - I(A:B)$. 
    It is closely related to the entanglement of purification (EoP), which represents the entanglement cost $E_{LOq}$ required to prepare a state with negligible quantum communication~\cite{10.1063/1.1498001}.
    However, the direct connection between the Markov gap and entanglement distillation remains unclear.
    In this work, we explore the possible relationship between bound entanglement and the Markov gap by investigating Haar random states. 

    The reflected entropy is holographically dual to the minimal cross section of the entanglement wedge~\cite{dutta2021canonical}. 
    For the classical holographic state, it is conjectured that $S_R = 2 E_p$.
    However, they differ significantly for general quantum states.
    The EoP never increases under the discarding of quantum systems~\cite{PhysRevA.91.042323}, whereas a counterexample has been found for the reflected entropy $S_R$~\cite{PhysRevA.107.L050401}.
    The gap between (twice) the EoP and the mutual information will be termed the EoP gap in the following, i.e. $g(A:B) = 2E_p(A:B) - I(A:B)$.
    EoP gap $g$ vanishes for triangle states, in which tripartite entanglement is absent, while the Markov gap $h$ vanishes for the so-called sum of triangle state (SOTS)~\cite{PhysRevLett.126.120501}.

    The detailed introduction to these concepts is provided in Sec.~\ref{sec: preliminaries}.
    Our main results are summarized as follows.
    In Sec.~\ref{sec: triangle_state}, we show that a nonzero Markov gap is necessary condition for the bound entanglement, revealing a possible relation between bound entanglement and the Markov gap.
    Conversely, one may ask whether the nonzero Markov gap is also sufficient for bound entanglement, or at least whether the distillability of entanglement can be characterized in terms of the Markov gap. 
    Therefore, In Secs.~\ref{sec: tripartite_entanglement} and~\ref{sec: Markov_gap}, we study tripartite Haar random states by employing the typicality of conditional mutual information (CMI) and Markov gap, respectively.    
    In the thermodynamic limit, the nonzero Markov gap typically exhibits two regimes. 
    One regime is strongly nonzero, corresponding to a large magnitude. 
    The other is weakly nonzero, exponentially small but still nonvanishing
    States with a weakly nonzero Markov gap typically possess PPT marginal states, i.e., either separable or bound entangled, whereas states with a strongly nonzero Markov gap typically possess NPT marginal states.
    Our findings suggest the relationship between the Markov gap and the entanglement distillation.

    In Sec.~\ref{sec: holography}, we interpret our results from the perspective of holographic duality.
    We show that unless $g = 0$, states with a vanishing Markov gap violates the monogamy of mutual information, a condition required for the geometric contribution of holographic entanglement entropy~\cite{PhysRevD.87.046003}.
    This supports the conjecture $S_R = 2E_p$, and indicates that the tripartite entanglement with a vanishing Markov gap perturbatively contributes to bulk corrections.
    In contrast, a strongly nonzero Markov gap corresponds to geometric contributions, while a weakly nonzero Markov gap corresponds to non-perturbative effects, where the PPT bound entanglement typically emerges.
    We hope that these results will stimulate further investigation into the role of the Markov gap in understanding entanglement distillation.

\section{Preliminaries} \label{sec: preliminaries}

    In this section, we introduce and review the physical notions and definitions used in this work.
    For the probabilistic adverb used in this paper and for more details, see Appendix~\ref{app: additional}. 

    \subsection{Entanglement and entanglement measures}

        Entanglement is a quantum property that is invariant under local unitary operations and does not increase under local operations and classical channels (LOCC)~\cite{RevModPhys.81.865}.
        A bipartite pure state is entangled, if the entanglement entropy of $\ket{\psi}_{AB}$ does not vanish
        \begin{equation}
            S(A) \equiv S(\rho_A) \equiv \mathrm{Tr}(\rho_A\log \rho_A) = \sum_i p_i \log p_i > 0.
        \end{equation}
        Here, the base of the logarithm is $2$ throughout this work.
        Since any pure bipartite entangled state can be reversibly distilled into maximally entangled states~\cite{PhysRevA.53.2046} via LOCC, the entanglement entropy serves as the unique measure of pure bipartite entanglement, analogous to thermodynamic entropy.

        In this work, we mainly focus on the bipartite entanglement of the mixed state, which corresponds to tripartite pure states through purification.
        The case of multipartite entanglement is more complicated, for a brief introduction, see Appendix~\ref{app: additional}.
        Entanglement entropy is insufficient to quantify entanglement in mixed states, and no single universal measure exists.
        One important criterion for detecting mixed-state entanglement is the violation of positive partial transpose (PPT) condition~\cite{PhysRevLett.77.1413}.
        The partial transpose of separable state $\rho = \sum_i p_i \rho_A^i \otimes \rho_B^i$ remains positive, $\rho^{T_B} = \sum_i p_i \rho_A^i \otimes (\rho_B^i)^{T_B} \geq 0$.
        Hence, a negative partial transpose (NPT) state $\rho_{AB}^{T_B} \not\geq 0$ is entangled.
        The negativity and logarithmic negativity
        \begin{equation}
            N(\rho) =  (\Vert \rho_{AB}^{T_B} \Vert_1-1)/2, \quad E_N(\rho) = \log \Vert \rho_{AB}^{T_B} \Vert_1,
        \end{equation}
        are entanglement measures of PPT: PPT states satisfy $N(\rho),E_N(\rho) = 0$, whereas NPT states yield $N(\rho), E_N(\rho) > 0$.
        PPT entangled states are bound entangled, meaning that they cannot be distilled into the maximally entangled state~\cite{PhysRevLett.80.5239}.
        Interestingly, generating bound entanglement requires maximally entangled states, reflecting the irreversibility of mixed-state entanglement~\cite{PhysRevLett.86.5803}.
        In addition, bound entangled states can be systematically constructed from the unextendible product bases~\cite{PhysRevLett.82.5385}.
        Another separability criterion is the reductive criterion~\cite{PhysRevA.59.4206}, which states that for separable states $\rho_{AB}$
        \begin{equation}
            \mathcal{I}_A \otimes \Lambda_B(\rho_{AB}) \geq 0, \quad
            \Lambda_A \otimes \mathcal{I}_B(\rho_{AB}) \geq 0
        \end{equation}
        with the map 
        \begin{equation}
            \Lambda(\rho) = I \mathrm{Tr} \rho - \sigma.
        \end{equation}
        A state that violates this criterion is thus entangled and, in particular, distillable~\cite{PhysRevA.59.4206}.  
        Figure~\ref{fig: figure1} schematically illustrates the set of bipartite quantum states.

        \begin{figure}[t]
            \centering
            \includegraphics[width=0.45\textwidth]{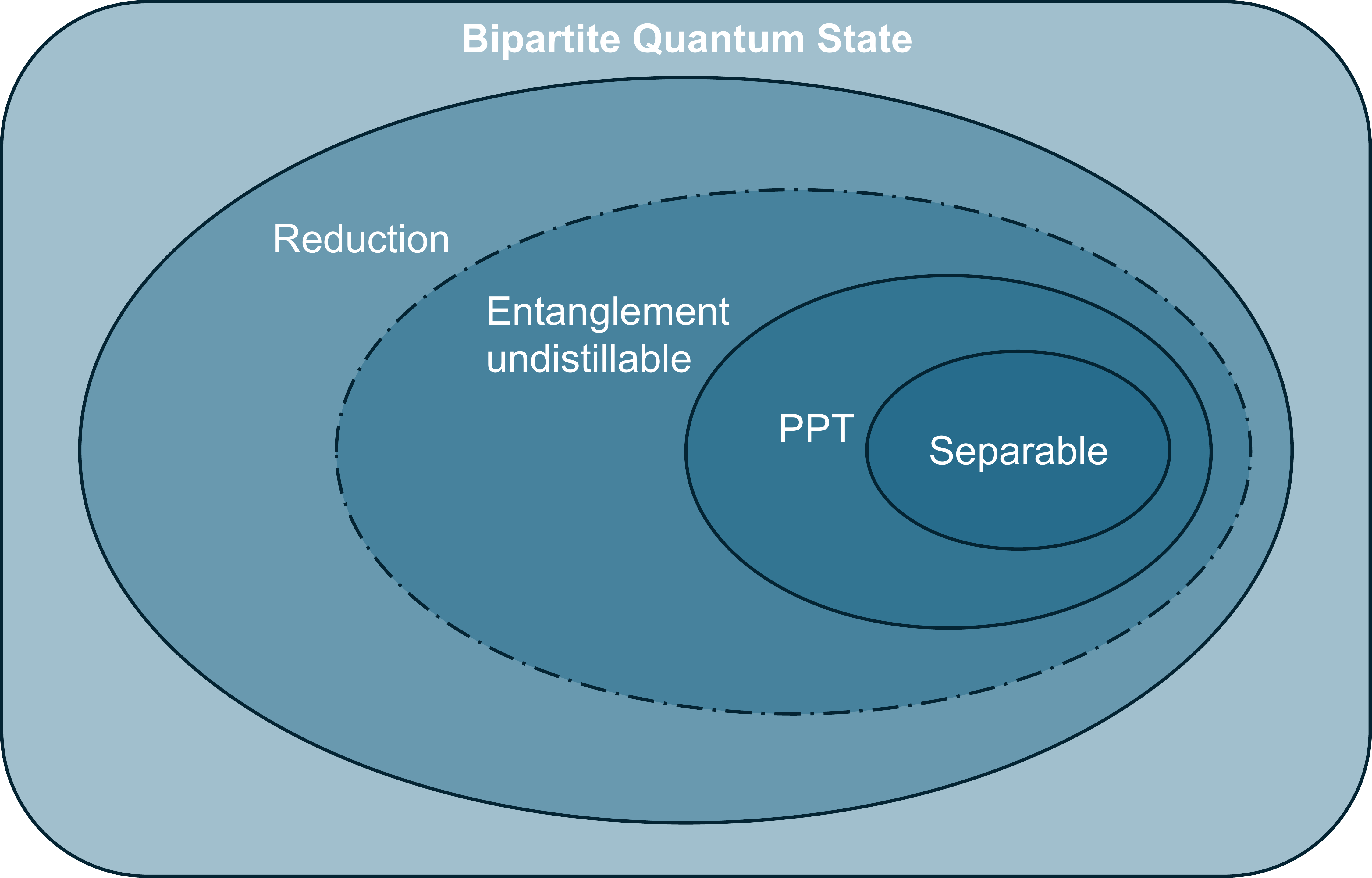}
            \caption{
                A diagram of the sets of bipartite quantum state.
            }
            \label{fig: figure1}
        \end{figure}

        EoP is a measure of total correlations, which is defined as the minimal entanglement entropy of the marginal state $\rho_{AC_1}$ over all possible purification $\ket{\psi_{ABC}}$, i.e. $\mathrm{Tr}_{C}\ket{\psi}\bra{\psi}_{ABC} = \rho_{AB}$ and all possible decompositions of system $C$ into $C_1$ and $C_2$
        \begin{equation}
            E_p(A:B) = \min S(\rho_{AC_1}).
        \end{equation}
        It is lower bounded by half of the mutual information, i.e., $E_p(A:B) \geq I(A:B)/2$, where
        \begin{equation}
            I(A:B) = S(A) + S(B) - S(AB).
        \end{equation}
        Furthermore, EoP is never increasing under the discarding of quantum system~\cite{PhysRevA.91.042323}.
        The asymptotic regularization of the EoP 
        \begin{equation}
            E_p^{\infty}(A:B) = \lim_{n\rightarrow\infty} \frac{1}{n} E_p(\rho_{AB}^{\otimes n})
        \end{equation}
        equals the entanglement cost $E_{LOq}$ under local operations and negligible communication (LOq) ~\cite{10.1063/1.1498001}
        \begin{align}
            & E_{LOq}(\rho_{AB})  \\
            & = \lim_{\epsilon\rightarrow 0} \inf \left\{\frac{m}{n}: d(\Lambda_{\mathrm{LOq}}(\ket{\phi}\bra{\phi}^{\otimes m}), \rho_{AB}^{\otimes n}) \leq \epsilon \right\}. \nonumber
        \end{align}
        Here, $\ket{\phi}$ is a maximally entangled state, $d(\cdot,\cdot)$ is the Bures distance, and $\Lambda_{LOq}$ is an LOq operation.
        This equivalence provides an operational interpretation for the EoP.

        The EoP gap, defined as  
        \begin{equation} \label{eq: g}
            g(A:B) \equiv 2 E_{P}(A:B) - I(A:B) ,
        \end{equation}
        vanishes, i.e., $g(A:B) = 0$, if and only if the tripartite pure state $\ket{\psi_{ABC}}$ is a triangle state~\cite{PhysRevLett.126.120501}, also known as a $2$-producible state~\cite{guhne2005multipartite,balasubramanian2014multiboundary}.
        This means the state can be expressed as a tensor product of states involving at most bipartite entangled [see Fig.~\ref{fig: figure2}(a)]
        \begin{equation} \label{eq: triangle_state}
            \ket{\psi_{ABC}} = \ket{\psi_{A_1 B_2}} \otimes \ket{\psi_{B_1 C_2}} \otimes \ket{\psi_{C_1 A_2}}.
        \end{equation}
        This structure implies a decomposition for each local Hilbert space $\mathcal{H}_{\alpha} = (\mathcal{H}_{\alpha_1} \otimes \mathcal{H}_{\alpha_2}) \oplus \mathcal{H}_{\alpha}^0$, (for $\alpha = A, B, C$), where $\ket{\psi_{\alpha_1\beta_2}} \in \mathcal{H}_{\alpha_1} \otimes \mathcal{H}_{\beta_2}$ for $\alpha, \beta = A, B, C$.
        This vanishing of EoP gap thus serve as a quantitative indicator for the absence of genuine tripartite entanglement in the tripartite system.

        Another related measure is the reflected entropy $S_R(A:B)$~\cite{dutta2021canonical}.
        It is defined as the entanglement entropy $S(\rho_{A\bar{A}})$ of the canonical purification $\ket{\psi}_{AB\bar{A}\bar{B}} = \ket{\sqrt{\rho_{AB}}} = \rho_{AB}^{1/2} \ket{\Phi}_{A\bar{A}} \ket{\Phi}_{B\bar{B}}$, 
        \begin{equation}
            S_R(A:B) = S_{A\bar{A}}(\ket{\psi}_{AB\bar{A}\bar{B}}) = S_{A\bar{A}}(\rho_{AB}^{1/2} \ket{\Phi}_{A\bar{A}}  \ket{\Phi}_{B\bar{B}}).
        \end{equation}
        Here, $\ket{\Phi}_{X\bar{X}} = \sum_i \ket{i}_X \ket{i}_{\bar{X}}$ is the unnormalized maximally entangled state between $X$ and its replica $\bar{X}$, for $X = A, B$.
        The reflected entropy is also lower bounded by mutual information, $S_R(A:B) \geq I(A:B)$, and is conjectured to be twice the EoP, $S_R(A:B) = 2 E_p(A:B)$, for classical holographic state~\cite{dutta2021canonical}.

        The Markov gap $h$, defined as
        \begin{equation}
            h(A: B) \equiv S_{R}(A:B) - I(A:B) = 0,
        \end{equation}
        vanishes $h(A:B) = 0$, if and only if the tripartite pure state $\ket{\psi_{ABC}}$ is a sum of triangle state (SOTS)~\cite{PhysRevLett.126.120501}, Fig.~\ref{fig: figure2}(b),
        \begin{equation}
            \ket{\psi_{ABC}} = \sum_l \sqrt{p_l} \ket{\psi_{A_1B_2}^l} \otimes \ket{\psi_{B_1C_2}^l} \otimes \ket{\psi_{C_1A_2}^l}.
        \end{equation}
        This structure corresponds to a decomposition for each local Hilbert space $\mathcal{H}_{\alpha} = \bigoplus_{l}(\mathcal{H}_{\alpha_1}^l \otimes \mathcal{H}_{\alpha_2}^l) \oplus \mathcal{H}_{\alpha}^0$, (for $\alpha = A, B, C$), with $\ket{\psi_{\alpha_1\beta_2}^l} \in \mathcal{H}_{\alpha_1}^l \otimes \mathcal{H}_{\beta_2}^l$.
        The Markov gap equals the CMI $I(\bar{A}:B|A)$, where the CMI is defined as 
        \begin{equation}
            I(A:C|B) = S(AB) + S(BC) - S(B) -S(ABC).
        \end{equation}
        This quantity is non-negative due to the strongly subadditivity of entropy and is related to the quantum Markov recovery property~\cite{hayden2004structure,fawzi2015quantum,PhysRevLett.115.050501,doi:10.1098/rspa.2015.0623,sutter2018approximate}
        \begin{equation} \label{eq: markov_recovery}
            I(A:C|B) \geq - \max_{\mathcal{R}_{B\rightarrow BC}} \log F(\rho_{ABC}, \mathcal{R}_{B\rightarrow BC}(\rho_{AB})).
        \end{equation} 
        Here, $\rho_{AB}$ is the marginal state of $\rho_{ABC}$, and $\mathcal{R}_{B\rightarrow BC}$ is a universal recovery map independent of $\rho_A$.
        Consequently, the Markov gap is related to the fidelity of a partial Markov recovery problem~\cite{hayden2021markov},
        \begin{align}
            h(A:B) & \geq - \max_{\mathcal{R}_{A\rightarrow A \bar{A}}} \log F(\rho_{AB\bar{A}}, \mathcal{R}_{A\rightarrow A \bar{A}}(\rho_{AB})), \\
            h(A:B) & \geq - \max_{\mathcal{R}_{B\rightarrow B \bar{B}}} \log F(\rho_{AB\bar{B}}, \mathcal{R}_{B\rightarrow B \bar{B}}(\rho_{AB})),
        \end{align}
        where $\rho_{AB\bar{B}}$ and $\rho_{AB\bar{A}}$ are marginal states of the canonical purification $\rho_{AB\bar{A}\bar{B}}$ of $\rho_{AB}$.
        
        \begin{figure}[t]
            \centering
            \includegraphics[width=0.48\textwidth]{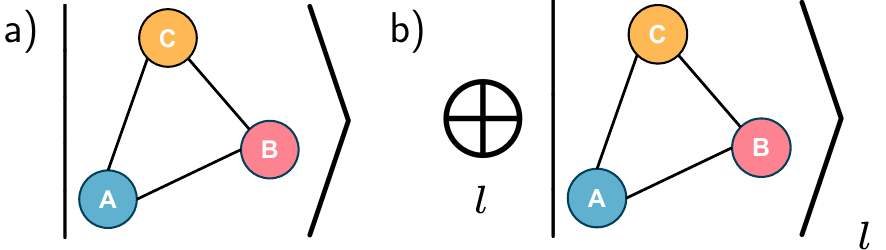}
            \caption{Diagrams of a) the triangle states and b) the sum of triangle states.}
            \label{fig: figure2}
        \end{figure}

    \begin{figure*}[t]
        \centering
        \includegraphics[width=0.9\textwidth]{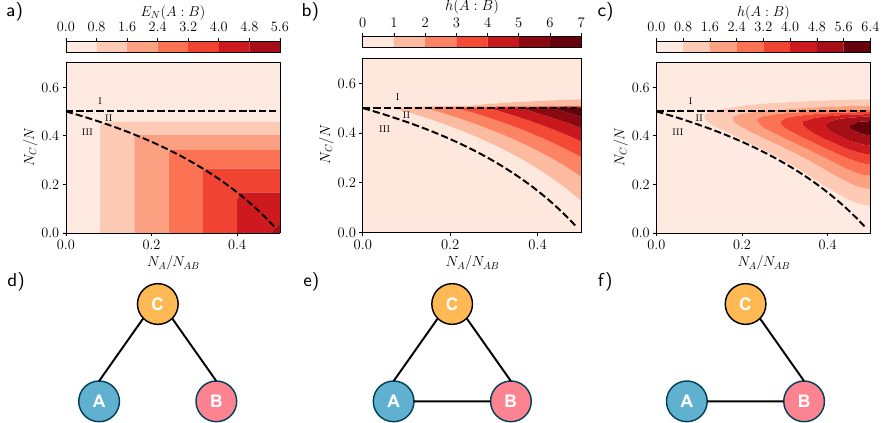}
        \caption{
            (a) The phase diagram of PPT entanglement phase transition in average logarithmic negativity $E_N(A:B)$ of Haar random state with fixed $N_{AB} = 10$ up to leading order.
            The dash lines distinguish the three phases of Haar ensemble, where region~I is PPT phase, region~II is entanglement saturation (ES) phase, and region~III is maximal entanglement (ME) phase.
            Assuming $N_A<N_B$, only the region $N_{A}/N_{AB}\leq 0.5$ is displayed.
            The region $N_{A}/N_{AB}\geq 0.5$ is symmetric to this region.
            (b) The average Markov gap $h(A:B)$ of Haar random state with fixed $N_{AB} = 10$ up to leading order.
            (c) The numerical average Markov gap $h(A:B)$ of Haar random state with fixed $N_{AB} = 10$, where $200$ instances are sampled for each point in the panel.
            (d)-(f) The schematic diagrams of the three phases, PPT, ES, and ME correspondingly, with the entanglement model of Bell pairs, where the solid lines between systems represent Bell pairs.
            The base of logarithm is $2$ in this figure.
        }
        \label{fig: figure3}
    \end{figure*}
    
    \subsection{Haar random state and its entanglement properties}

    A Haar random state is an ensemble of states randomly sampled from Hilbert space according to Haar measure, the uniform probability measure.
    It has wide applications in quantum physics, including quantum computation and quantum communication~\cite{PhysRevA.55.1613,PhysRevA.60.1888,PhysRevA.80.012304,PhysRevLett.106.180504,PhysRevX.6.041044,boixo2018characterizing,arute2019quantum}, information scrambling and quantum chaos~\cite{PhysRevB.98.064309,PhysRevD.98.086026,sekino2008fast}, as well as black hole physics~\cite{sekino2008fast,Hayden_2007,shenker2014black,kudler2022negativity}.
    A Haar random state can be represented as an ensemble of state $\{\hat{U}\ket{\psi_0}\}$, where $\hat{U}$ is a random unitary operator distributed in the Haar measure on the unitary group $U(\mathcal{H})$, and $\ket{\psi_0}$ is an arbitrary fixed pure state.
    The Haar measure is the unique probability measure on the unitary group that is both left-invariant and right-invariant~\cite{watrous2018theory}
    \begin{align}
        &\mathbb{E}_{U}[1] =  \int_{\mathrm{Haar}} dU = 1, \\
        &\mathbb{E}_{VU}[f(U)] = \mathbb{E}_{UV}[f(U)] = \mathbb{E}_{U}[f(U)],
    \end{align}
    where $U, V \in U(\mathcal{H})$ are the unitary matrices on the Hilbert space $\mathcal{H}$.
    In high dimension, the individual instances of Haar random states exhibit similar properties, which is a general phenomenon arising from the concentration of the measure~\cite{ledoux2001concentration}.

    For a bipartite Haar random state $\ket{\psi_{AB}}$ with system dimensions $D_A$ and $D_B$ respectively, the entanglement entropy almost saturates the maximum value. This is captured by Page's formula~\cite{PhysRevLett.71.1291}
    \begin{equation}
        \bar{S}_A \equiv \bar{S}(\rho_A) = \log \min(D_{A},D_B) - \frac{\min(D_A,D_B)}{2 \max(D_A,D_B)}.
    \end{equation}
    For a tripartite Haar random state $\ket{\psi_{ABC}}$ with system dimensions $D_A, D_B$ and $D_C$, the marginal state $\rho_{AB}$ is a Haar induced state governed by the induced measure $\nu_{D_{AB}, D_C}$, where $D_{AB}$ and $D_C$ are the dimensions of systems $AB$ and $C$, respectively.
    For further details on the induced measure, see Appendix~\ref{app: additional}.

    Threshold phenomena~\cite{collins2016random} have been found for the marginal state $\rho_{AB}$ of Haar random state.
    Specifically, there exist a threshold function $s_0(D)$, such that: \textit{(a)} if $D_C \leq (1-\epsilon) s_0$, the Haar induced state $\rho_{AB}$ lacks property $X_D$ in probability 
    \begin{equation}
        \lim_{D\rightarrow\infty}\mathbb{P}_{\nu_{D_{AB},D_C}}(\rho \in X_D) = 0;
    \end{equation}
    conversely, \textit{(b)} if $D_C \geq (1+\epsilon) s_0$, the Haar induced state $\rho_{AB}$ possess the property $X_D$ in probability
    \begin{equation}
        \lim_{D\rightarrow\infty}\mathbb{P}_{\nu_{D_{AB},D_C}}(\rho \in X_D) = 1.
    \end{equation}
    Here, $X_D$ denotes a set of states with a specific property, such as separability ($\mathrm{SEP}_{D_A,D_B}$) or positive partial transpose ($\mathrm{PPT}_{D_A,D_B}$). 
    The threshold for separability $s_{\mathrm{SEP}}$ satisfies~\cite{PhysRevA.85.030302,aubrun2014entanglement}
    \begin{align}
        & c D_A D_B \min(D_A,D_B) \leq s_{\mathrm{SEP}} \\
        &~~~~ \leq C D_A D_B \min(D_A,D_B) \log^2(D_AD_B), \nonumber 
    \end{align}
    and the threshold for PPT $s_{\mathrm{PPT}}$ is given by~\cite{PhysRevA.85.030302,PhysRevA.85.062331}
    \begin{equation}
        s_{\mathrm{PPT}} = 4 D_A D_B.
    \end{equation}
    For a system of qubits, where $D_{A,B,C} = 2^{N_{A,B,C}}$, in the thermodynamic $N \equiv N_A + N_B + N_C \rightarrow \infty$ with fixed proportions $n_{A,B,C} = {N_{A,B,C}}/{N}$, the thresholds for separability and PPT become
    \begin{equation}
    n_{\mathrm{SEP}} = [1 + \min(n_A, n_B)]/2 \leq {3}/{5}, \quad
    n_{\mathrm{PPT}} = {1}/{2}.
    \end{equation} 
    Consequently, the marginal state $\rho_{AB}$ is typically NPT entangled when $n_{C}<s_{\mathrm{PPT}}$, PPT entangled when $n_{\mathrm{PPT}}<n_{C}<n_{\mathrm{SEP}}$, and separable when $n_{C}>n_{\mathrm{SEP}}$.  

    The threshold for PPT divides the marginal state $\rho_{AB}$ into three phases~\cite{PRXQuantum.2.030347}.
    The phase diagram, shown in Fig.~\ref{fig: figure3}(a), have been observed in experiment~\cite{liu2023observation}.
    Note that the critical point for these phases is $n_{\mathrm{PPT}} = 1/2$, thus the tripartite pure states $\ket{\psi_{ABC}}$ in the PPT and maximal entanglement (ME) phases have the same entanglement structure up to a relabeling of the subsystems. 

    Reference~\cite{PRXQuantum.2.030347} employed a phenomenological model where the entanglement between the systems is assumed to consist of tensor products of Bell pairs. 
    This model the phase diagram predicts the phase diagram using Page's formula at leading order. 
    According to this model, there is no Bell pair between the systems $A$ and $B$ in the PPT phase, no Bell pair between the systems $A$ and $C$ in the ME phase, and there are Bell pairs between every pair of systems in the ES phase.
    Schematic diagrams of these three phases are shown in Fig.~\ref{fig: figure3}(d)-(f), respectively.

    \subsection{Brief introduction to holographic duality}

    In holographic duality, the entropy of a holographic state $\rho_A$ is given by the generalized entropy~\cite{faulkner2013quantum,engelhardt2015quantum}, up to non-perturbative corrections $\propto O(e^{-1/G_N})$
    \begin{equation}
        S(A) = \min_{\gamma_A} S_{\mathrm{gen}}(\gamma_A) \equiv \min_{\gamma_A} \left[\frac{\mathcal{A}(\gamma_{A})}{4G_N} + S_{\mathrm{bulk}}(\gamma_A)\right].
    \end{equation} 
    Here, $\mathcal{A}(\gamma_{A})$ is the area of the surfaces $\gamma_{A}$ in bulk that is homologous to the boundary region $A$ corresponding to the state, $G_N$ is the Newton constant of the bulk geometry, and $S_{\mathrm{bulk}}$ is the entanglement entropy of bulk region $\Omega$ bounded by $\partial \Omega = A\cup \gamma_{A}$.
    The minimal surface $\gamma_A$ is called the quantum extremal surface~\cite{engelhardt2015quantum}.
    In semiclassical limit, it approaches the classical extremal surface, known as the Ryu-Takayanagi (RT) surface $\gamma_{A,\mathrm{RT}}$~\cite{PhysRevLett.96.181602,Ryu_2006}.
    Perturbatively, the entropy is expressed as
    \begin{equation}
        S(A) = \frac{\mathcal{A}(\gamma_{A,\mathrm{RT}})}{4G_N} + S_{\mathrm{bulk}}(\gamma_{A,\mathrm{RT}}),
    \end{equation}    
    where the Ryu-Takayanagi (RT) surface $\gamma_{A,\mathrm{RT}}$ is the minimal-area surface in the bulk geometry among all surfaces sharing the same boundary $\partial A$ as the boundary region $A$.
    The term $S_{\mathrm{bulk}} \propto O(1/G_N)$ represents the entanglement entropy of the bulk region $\Omega_{\mathrm{RT}}$ enclosed by the RT surface $\gamma_{A,\mathrm{RT}}$ and the boundary region $A$.

    The reflected entropy $S_R(A:B)$ is formulated as~\cite{dutta2021canonical}
    \begin{equation}
        S_R(A:B) = \min_{\sigma_{A:B}} \left[\frac{2 \mathcal{A}(\sigma_{A:B})}{4G_N} + S_{R,\mathrm{bulk}}(\sigma_{A:B})\right].
    \end{equation}
    Here, $\sigma_{A:B}$ is the surface that divides the entanglement wedge $W(A:B)$, i.e., the bulk region bounded by the boundary region $AB$ and its RT surfaces $\gamma_{AB,\mathrm{RT}}$, into two regions through which the surface $\sigma_{A:B}$ is homologous to the boundary region $A$ or $B$, respectively.
    $S_{R,\mathrm{bulk}}$ is the reflected entropy between the two bulk regions separated by the surface $\sigma_{A:B}$.
    The minimal surface $\sigma_{A:B}$ is called the entanglement wedge cross section.
    The Markov gap is given by
    \begin{align}
        &h(A:B) = \frac{\mathcal{A}(\gamma_{A,\mathrm{KRT}}) - \mathcal{A}(\gamma_{A,\mathrm{RT}})}{4G_N} \nonumber \\
        &~~~~+ \frac{\mathcal{A}(\gamma_{B,\mathrm{KRT}}) - \mathcal{A}(\gamma_{B,\mathrm{RT}})}{4G_N} + h_{\mathrm{bulk}}(A:B),
    \end{align}
    where $\gamma_{A,\mathrm{KRT}}$ is the kicked-Ryu-Takayanagi (KRT) surface~\cite{hayden2021markov}, and $h_{\mathrm{bulk}}(A:B)\propto O(1/G_N) $ denotes the bulk contribution to the Markov gap.
    The KRT surface $\gamma_{A,\mathrm{KRT}}$ is formed by combining the entanglement wedge cross section $\sigma_{A:B}$ with a portion of the RT surfaces $\gamma_{AB,\mathrm{RT}}$ of the boundary region $AB$, such that the union $\gamma_{A,\mathrm{KRT}}$ is homologous to the RT surface $\gamma_{A,\mathrm{RT}}$ and the boundary region $A$.
    A diagram illustrating these surfaces is provided in Fig.~\ref{fig: figure4}.

    \begin{figure}
        \centering
        \includegraphics[width=0.4\textwidth]{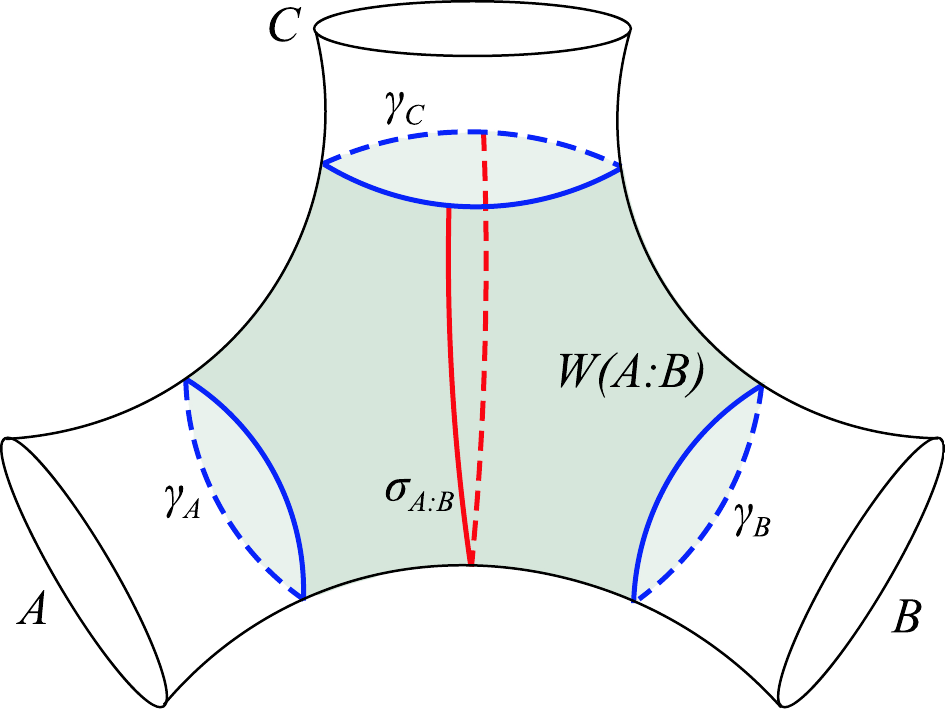}
        \caption{A diagram of a three-boundary wormhole. $\gamma_{A}$, $\gamma_{B}$, and $\gamma_{AB}$ are RT surfaces of $A$, $B$, and $AB$, respectively.
            Entanglement wedge $W(A:B)$ is the bulk region between RT surfaces $\gamma_{A}$, $\gamma_{B}$, and $\gamma_{AB}$. 
            $\sigma_{A:B}$ is the entanglement wedge cross section of $A$ and $B$, which divides $\gamma_{AB}$ into two parts.
            The union of $\sigma_{A:B}$ and one part of $\gamma_{AB}$ is one of the KRT surfaces.}
        \label{fig: figure4}
    \end{figure}

\section{Bound Entanglement Has Nonzero Markov Gap} \label{sec: triangle_state}

    The mean result of this section is the following theorem.
    \begin{theorem} \label{theorem: bound_entangled} 
        The marginal state $\rho_{AB}$ of SOTS is distillable if entangled, i.e. not a bound entangled state.
    \end{theorem} \noindent
    Since Markov gap vanishes if and only if the canonical purification is SOTS~\cite{PhysRevLett.126.120501}, the bound entangled state thus must have a nonzero Markov gap $h(A:B) > 0$.
    To prove this, we require the following lemma.
    \begin{lemma}\label{lemma: SOTS_separable} 
        For a sum of triangle states of tripartite system $ABC$, the marginal state 
        \begin{equation}
            \rho_{AB} = \sum_l p_l \ket{\psi_{A_1 B_2}^l}\!\bra{\psi_{A_1 B_2}^l} \otimes \rho_{B_1}^l \otimes \rho_{A_2}^l,    
        \end{equation}
        is separable if and only if each $\ket{\psi_{A_1 B_2}^l}$ is separable, and consequently, 
        \begin{align} \label{eq: SOTS_separable}
            S_{A:B} & \equiv \sum_l p_l S_{A_1}(\ket{\psi_{A_1 B_2}^l})  \\
            & = \frac{1}{2}[I(A:B) - g(A:B)] \nonumber \\
            & = E_p(A:B) - g(A: B) = 0. \nonumber
        \end{align}
    \end{lemma}\noindent
    Here, Eq.~(\ref{eq: SOTS_separable}) defines the quantity $S_{A:B}$, and the second and third equalities are follow from Proposition~\ref{proposition: SOTS_g}.
    \begin{proof}
        Assume this marginal state is separable.
        The projector $\hat{P}_{l}^{(AB)}$ onto the Hilbert space $\mathcal{H}_{A}^l \otimes \mathcal{H}_B^l$ is a local operator
        \begin{equation}
            \hat{P}_{l}^{(AB)} = \hat{P}_{l}^{(A_1)} \otimes \hat{P}_{l}^{(A_2)} \otimes \hat{P}_{l}^{(B_1)} \otimes \hat{P}_{l}^{(B_2)}.
        \end{equation}
        Projection therefore does not generate entanglement, so, the projected state is still a separable state
        \begin{equation}
            \hat{P}_{l}^{(AB)}\rho_{AB}\hat{P}_{l}^{(AB)} = p_l \ket{\psi_{A_1 B_2}^l}\!\bra{\psi_{A_1 B_2}^l} \otimes \rho_{B_1}^l \otimes \rho_{A_2}^l.
        \end{equation}
        It follows that each $\ket{\psi_{A_1 B_2}^l}$ is separable.
        Conversely, if each $\ket{\psi_{A_1 B_2}^l}$ is separable, the marginal state is separable by definition.

        Furthermore, each $\ket{\psi_{A_1 B_2}^l}$ is separable if and only if $S_{A_1}(\ket{\psi_{A_1 B_2}^l}) = 0$, which implies  $S_{A:B} \equiv \sum_l p_l S_{A_1}(\ket{\psi_{A_1 B_2}^l}) = 0$.
        Direct calculation shows that (see Proposition~\ref{proposition: SOTS_g})
        \begin{align}
            S_{A:B} & = \frac{1}{2}[I(A:B) - g(A:B)] \\
            & = E_p(A:B) - g(A: B) \nonumber
        \end{align}
    \end{proof}

    \begin{proof}[Proof of Theorem~\ref{theorem: bound_entangled}]
        Consider a SOTS
        \begin{equation}
            \ket{\psi_{ABC}} = \sum_l \sqrt{p_l} \ket{\psi_{A_1B_2}^l} \otimes \ket{\psi_{B_1C_2}^l} \otimes \ket{\psi_{C_1A_2}^l}.
        \end{equation}
        Since the Hilbert spaces $\mathcal{H}_{A}^l$, $\mathcal{H}_{B}^l$, and $\mathcal{H}_{C}^l$ are orthogonal for different index $l$, the marginal state is
        \begin{equation}
            \rho_{AB} = \sum_l p_l \ket{\psi_{A_1 B_2}^l}\!\bra{\psi_{A_1 B_2}^l} \otimes \rho_{B_1}^l \otimes \rho_{A_2}^l.  
        \end{equation}
        Since the projection $\hat{P}_A^l\otimes\hat{P}_B^l$ on the space $\mathcal{H}_{A}^l \otimes \mathcal{H}_{B}^l$ is a local operation, the projected state  
        \begin{equation}
            \ket{\psi_{A_1 B_2}^l}\!\bra{\psi_{A_1 B_2}^l} \otimes \rho_{B_1}^l \otimes \rho_{A_2}^l
        \end{equation}
        can be extracted from the marginal state via LOCC.
        Moreover, Since any pure bipartite entangled state is distillable~\cite{PhysRevA.53.2046}, Bell states can be distilled from this projected state if it is entangled.
        According to Proposition~\ref{lemma: SOTS_separable}, the marginal state $\rho_{AB}$ is entangled if and only if at least one of these projected state is entangled, and thus is distillable.
    \end{proof}

    Furthermore, a SOTS has a vanishing Markov gap $h(A:B) = 0$.
    In the special case of a triangle state, the EoP gap also vanishes $g(A:B) = 0$.
    Here, we are interested in the value of $g(A:B)$ for a general SOTS, where it is non-vanishing.
    Given the orthogonality of $(\mathcal{H}_{\alpha_1}^l \otimes \mathcal{H}_{\alpha_2}^l)$, $g(A:B)$ for a SOTS can be calculated explicitly. 
    \begin{proposition} \label{proposition: SOTS_g} 
        For a sum of the triangle state 
        \begin{equation}
            \ket{\psi_{ABC}} = \sum_l \sqrt{p_l} \ket{\psi_{A_1 B_2}^l} \ket{\psi_{B_1 C_2}^l} \ket{\psi_{C_1 A_2}^l},
        \end{equation}
        the EoP gap is given by
        \begin{align} \label{eq: SOTS_g}
            g(A:B) & = g(B:C) = g(A:C) \\
            & = H(p_l) \equiv -\sum_l p_l \log p_l. \nonumber
        \end{align}
    \end{proposition}

    \begin{proof}
        The EoP is defined as
        \begin{equation}
            E_p(A:B) = \min_{\hat{U}_C} S_{AC_1}(\hat{U}_C \ket{\psi_{ABC}}).
        \end{equation}
        Since the Hilbert spaces $\mathcal{H}_{X}^l$ are orthogonal, the marginal state of $\hat{U}_C \ket{\psi_{ABC}}$ is 
        \begin{equation}
            \tilde{\rho}_{AC_1} =  \sum_l p_l \rho_{A_1}^l \otimes \mathcal{N}^l(\ket{\psi_{C_1 A_2}^l}\!\bra{\psi_{C_1 A_2}^l}),
        \end{equation}
        where $\mathcal{N}^l(\cdot) = \mathrm{Tr}_{B_1C_2}[\hat{U}_C(\cdot \otimes \ket{\psi_{B_1 C_2}^l}\!\bra{\psi_{B_1 C_2}^l})\hat{U}_C^{\dagger}]$.
        The entropy is
        \begin{align}
            & S_{AC_1}(\hat{U}_C \ket{\psi_{ABC}}) = \sum_l p_l[S_{A_1}(\ket{\psi_{A_1 B_2}^l}) \\
            &~~~~~~~~+ S_{C_1 A_2}( \mathcal{N}^l(\ket{\psi_{C_1 A_2}^l}\!\bra{\psi_{C_1 A_2}^l}))] + H(p_l). \nonumber
        \end{align}
        where 
        \begin{equation}
            S_{C_1 A_2}( \mathcal{N}^l(\ket{\psi_{C_1 A_2}^l}\!\bra{\psi_{C_1 A_2}^l})) \geq 0,
        \end{equation}
        and the equality holds when $\hat{U}_C = \hat{I}$.
        Therefore, the EoP is minimizied at $\hat{U}_C = \hat{I}$, yielding
        \begin{equation}
            E_p(A:B) = \sum_l p_l S_{A_1}(\ket{\psi_{A_1 B_2}^l}) + H(p_l).
        \end{equation}

        Then, we calculate the mutual information
        \begin{equation}
            I(A:B) = S_A + S_B - S_{AB}.
        \end{equation} 
        The relevant density matrices are
        \begin{align}
            \rho_{AB} & = \sum_l p_l \ket{\psi_{A_1 B_2}^l}\!\bra{\psi_{A_1 B_2}^l} \otimes \rho_{B_1}^l \otimes \rho_{A_2}^l, \\
            \rho_{A} & = \sum_l p_l \rho_{A_1}^l \otimes \rho_{A_2}^l, \\
            \rho_{B} & = \sum_l p_l \rho_{B_2}^l \otimes \rho_{B_1}^l.
        \end{align}
        Their entropies are 
        \begin{align}
            S_A & = \sum_l p_l [S_{A_1}(\ket{\psi_{A_1 B_2}^l}) + S_{A_2}(\ket{\psi_{C_1 A_2}^l})] + H(p_l), \\
            S_B & = \sum_l p_l [S_{B_2}(\ket{\psi_{A_1 B_2}^l}) + S_{B_1}(\ket{\psi_{B_1 C_2}^l})] + H(p_l), \\
            S_{AB} & = \sum_l p_l [S_{A_2}(\ket{\psi_{C_1 A_2}^l}) + S_{B_1}(\ket{\psi_{B_1 C_2}^l})] + H(p_l).
        \end{align}
        Thus, the mutual information is
        \begin{equation}
            I(A:B) = 2\sum_l p_l S_{A_1}(\ket{\psi_{A_1 B_2}^l}) + H(p_l),
        \end{equation}
        and consequently,
        \begin{equation}
            g(A:B) = 2E_p(A:B) - I(A:B) = H(p_l).
        \end{equation}
        Due to the permutation symmetry among labels $A, B$ and $C$, it follows that
        \begin{equation}
            g(A:B) = g(B:C) = g(A:C) = H(p_l).
        \end{equation} 
    \end{proof} 
    From the calculation of Eq.~(\ref{eq: SOTS_g}) follows
    \begin{equation} \label{eq: SOTS_entropy}
        S_A = S_{A:B} + S_{A:C} + g(A:B).
    \end{equation}
    
    In contrast with the Haar random state, we now consider random stabilizer states, which form a $3$-design of Haar random state, as a counterpart example, and to illustrate above results.
    This comparison reveals that the bound entanglement properties of Haar random states differ from those of random stabilizer states.
    A random stabilizer state is locally equivalent to a tensor product of Bell states and GHZ states in a tripartite system~\cite{10.1063/1.2203431}
    \begin{align}
        \ket{\psi_{ABC}^{\mathrm{st}}} & = \mathcal{U}_A\mathcal{U}_B \mathcal{U}_C \ket{0}^{\otimes s_A}  \ket{0}^{\otimes s_B}  \ket{0}^{\otimes s_C}  \ket{\mathrm{GHZ}}^{\otimes g_{ABC}} \nonumber\\
        &~~~~ \otimes \ket{\mathrm{EPR}}^{\otimes e_{AB}} \ket{\mathrm{EPR}}^{\otimes e_{BC} } \ket{\mathrm{EPR}}^{\otimes e_{AC}},
    \end{align}
    and the logarithmic negativity of its marginal state $\rho_{AB}$ is given by the simple function
    \begin{equation}
        E_N = \frac{1}{2} \log(p_2^2/p_3) = e_{AB}\log 2.
    \end{equation}
    It is straightforward to see that the marginal state $\rho_{AB}$ of $\ket{\psi_{ABC}^{\mathrm{st}}}$ is separable if and only if 
    \begin{equation}
        e_{AB} = 0.
    \end{equation}
    This implies that for stabilizer states, separability is equivalent to the PPT condition. 
    Consequently, the thresholds for separability and PPT coincide for random stabilizer states.
    Since GHZ states are specific instances of SOTS, this observation is consistent with Theorem~\ref{theorem: bound_entangled}.
    
    On the other hand, the entropy of the stabilizer state $\ket{\psi_{ABC}^{\mathrm{st}}}$ for a subsystem $X = A,B,C$ is 
    \begin{equation} \label{eq: stabilizer_entropy}
        S_{X} =  \left(\sum_{Y\neq X} e_{XY} + g_{ABC}\right) \log 2,
    \end{equation}
    and from Eqs.~(\ref{eq: SOTS_g}) and~(\ref{eq: SOTS_separable}), we have
    \begin{align}
        S_{X:Y} & = e_{XY} \log 2, \\
        g(A:B) & = g_{ABC} \log 2,
    \end{align}
    which consistent with Eq.~(\ref{eq: SOTS_entropy}).
    Here, we observe that $S_{A:B} = E_N(\rho_{AB})$, which quantifies the numbers of Bell pairs between $A$ and $B$, while $g(A:B)$ measures the number of GHZ state shared among $A$, $B$, and $C$ in the stabilizer states~\cite{nguyen2018entanglement,PRXQuantum.2.030313}. 
    This suggests a possible general relation between $S_{A:B}$ and $E_N(\rho_{AB})$ in SOTS (see Proposition~\ref{proposition: negativity}). 
    Thus, the average bipartite entanglement of random stabilizer state between $A$ and $B$ is 
    \begin{equation}
        \bar{e}_{AB} \log 2 = \frac{1}{2}(\bar{S}_A + \bar{S}_B - \bar{S}_C - \bar{g}(A:B)),
    \end{equation}
    where $g = g(A:B) = g(B:C) = g(A:C)$. 
    It has also been shown that the average bipartite entanglement entropy of a random stabilizer states satisfies~\cite{PhysRevA.74.062314}
    \begin{equation}
        \bar{S}_X \geq \log D_{X\min} - \frac{D_{X\min}}{D_{X\max}}.
    \end{equation}
    Furthermore, using the upper bound on entropies $\bar{S}_X \leq \log D_{X\min}$, and the fact that average number of GHZ state is relatively small~\cite{PhysRevA.74.062314}, $\bar{g}_{ABC} \leq O(D^{-\alpha})$, the average number of Bell pairs between system $AB$ is 
    \begin{widetext}
        \begin{equation}
            \bar{e}_{AB}  \in \left\{\begin{array}{lc}
                \left[0, \frac{D_A D_B}{2\log 2 D_C}\right], & n_C > 1/2 \\ 
                \min\{N_A,N_B\}  + O(D^{-\alpha})\left[- C_1, C_2\right], & \max\{n_A,n_B\} > 1/2 \\
                \frac{1}{2}(N - 2 N_C) + O(D^{-\alpha})\left[- C_1, C_2\right], & \max\{n_A,n_B,n_C\} < 1/2,
            \end{array}\right.
        \end{equation}     
    \end{widetext}
    which is consistent with the Bell pair model, Fig.~\ref{fig: figure3}(d)~(e), at leading order.
    Therefore, a critical point exists for $\bar{e}_{AB}$ at $n_C = 1/2$, similar to the PPT threshold of Haar random states, dividing the phase diagram into three regions, i.e. PPT, ME, and ES phases. 

    In PPT phase, there is a region $n_{\mathrm{PPT}} < n_C < n_{\mathrm{SEP}}$ where Haar random states are typically bound entanglement.
    However, as shown above, random stabilizer states cannot be bound entangled. 
    In ES phase, the partial transpose spectrum of Haar random state is widely distributed around zero~\cite{PRXQuantum.2.030347,PhysRevA.109.012422}.
    This qualitative difference from the spectral distribution of random stabilizer states~\cite{PhysRevA.109.012422}, which concentrated on the $\pm \sqrt{p_3} \neq 0$, where $p_3$ is the third partial transpose moment, highlights the distinct entanglement properties of Haar random states.

    \begin{proposition} \label{proposition: negativity}
        The marginal state $\rho_{AB}$ of a sum of triangle state $\ket{\psi_{ABC}}$
        \begin{equation}
            \rho_{AB} = \sum_l p_l \ket{\psi_{A_1 B_2}^l}\!\bra{\psi_{A_1 B_2}^l} \otimes \rho_{B_1}^l \otimes \rho_{A_2}^l,    
        \end{equation}
        satisfies 
        \begin{equation}
            E_N(\rho_{AB}) \geq S_{A:B} = \frac{1}{2}[I(A:B) - g(A:B)].
        \end{equation}
        Equality holds if and only if the states $\ket{\psi_{A_1 B_2}^l}$ are equivalent to the same numbers of Bell pairs up to local unitary.
    \end{proposition}
    \begin{proof}
        The logarithmic negative satisfies the properties of concavity~\cite{PhysRevA.65.032314}
        \begin{equation}
            E_N\left(\sum_i p_i \rho_i\right) \geq \sum_i p_i E_N(\rho_i),
        \end{equation}
        additivity
        \begin{equation}
            E_N(\rho_1\otimes \rho_2) = E_N(\rho_1) + E_N(\rho_2),
        \end{equation}
        and for pure state $\rho_{AB} = \ket{\psi}\bra{\psi}$
        \begin{equation}
            E_N(\rho_{AB}) \geq S(A) = S(\rho_A),
        \end{equation}
        where $S(A)$ is the entanglement entropy of the pure state. 
        For the marginal state $\rho_{AB}$, the inequality follows from the above properties and the fact that $E_N = 0$ for separable state $\rho_{B_1}^l \otimes \rho_{A_2}^l$.

        For the conditions of equality, consider the partial transpose 
        \begin{equation}
            \rho_{AB}^{T_B} = \sum_l p_l \ket{\psi_{A_1 B_2}^l}\!\bra{\psi_{A_1 B_2}^l}^{T_B} \otimes (\rho_{B_1}^l)^{T_B} \otimes \rho_{A_2}^l.
        \end{equation}
        Since the subspace $\mathcal{H}_{AB}^l$ are orthogonal, we have 
        \begin{equation}
            \Vert \rho_{AB}^{T_B} \Vert_1 = \sum_l p_l \Vert \ket{\psi_{A_1 B_2}^l}\!\bra{\psi_{A_1 B_2}^l}^{T_B} \Vert_1,
        \end{equation}
        and the logarithmic negativity
        \begin{align}
            E_N(\rho_{AB}) & = \log \Vert \rho_{AB}^{T_B} \Vert_1 \nonumber\\
            & = \log \sum_l p_l \exp E_N(\ket{\psi_{A_1 B_2}^l}) \nonumber \\
            & \geq \sum_l p_l E_N(\ket{\psi_{A_1 B_2}^l}) \nonumber \\
            & \geq \sum_l p_l S_{A}(\ket{\psi_{A_1 B_2}^l}).
        \end{align}
        Equality holds if and only if $E_N(\ket{\psi_{A_1 B_2}^l}) = E_N(\ket{\psi_{A_1 B_2}^{l'}}) = S_{A}(\ket{\psi_{A_1 B_2}^l})$ for different $l$ and $l'$.
        Note that $E_N(\ket{\psi_{A_1 B_2}^l}) = S_{A}(\ket{\psi_{A_1 B_2}^l})$ if and only if $\ket{\psi_{A_1 B_2}^l}$ are equivalent to maximal entangled up to local unitary.
    \end{proof}

\section{Tripartite Entanglement in Haar Random State} \label{sec: tripartite_entanglement}

    In tripartite state, tripartite entanglement is a necessary condition for the bound entanglement in the bipartite marginal states.
    Although the absence of tripartite entanglement can be characterized by a vanishing EoP gap, $g(A:B)= 0$, this quantity is generally difficult to compute.
    In the region $n_{\mathrm{PPT}} = 1/2 < n_C < n_{\mathrm{SEP}}$, the marginal state $\rho_{AB}$ of a Haar random state is typically bound entangled, denoted by $\mathbb{P}[\mathrm{BND}_{AB}] \rightarrow 1$. 
    Therefore, tripartite entanglement must be present 
    \begin{equation}
        \mathbb{P}[g>0] \geq \mathbb{P}[\mathrm{BND}_{AB}]\rightarrow 1.
    \end{equation}
    Furthermore, by Theorem~\ref{theorem: bound_entangled}, this also implies the existence of non-SOTS-type tripartite entanglement 
    \begin{equation}
        \mathbb{P}[h>0] \geq \mathbb{P}[\mathrm{BND}_{AB}]\rightarrow 1.
    \end{equation}
    In this section, we investigate the tripartite entanglement in Haar random state $\ket{\psi_{ABC}}$ by analyzing  the typicality of conditional mutual information.

    Our motivation stems from the following observation.
    The marginal state of a triangle state, Eq.~(\ref{eq: triangle_state}), is
    \begin{equation}
        \rho_{AB} = \ket{\psi_{A_1 B_2}}\bra{\psi_{A_1 B_2}} \otimes \rho_{B_1} \otimes \rho_{A_2},
    \end{equation}
    which is separable if and only if $\ket{\psi_{A_1 B_2}}$ is separable. 
    It requires that the tripartite 2-producible state is of the form 
    \begin{equation} \label{eq: Markov_chain}
        \ket{\psi_{ABC}} = \ket{\psi_{B C_1}} \otimes \ket{\psi_{C_2 A}}.
    \end{equation}
    This is equivalent to the condition that the conditional mutual information vanishes~\cite{hayden2004structure} for the tripartite state $\ket{\psi_{ABC}}$,
    \begin{equation}
        I(A:B|C) = I(A:B) = 0.
    \end{equation}
    Let $\mathrm{SEP}_{AB}$ denote the event that the marginal state $\rho_{AB}$ is separable.
    We then have
    \begin{equation} \label{eq: sep_g}
        \mathrm{SEP}_{AB} \cap \{g(A:B) = 0\} = \{I(A:B|C)=0\}.
    \end{equation}
    Therefore, the typicality of the conditional mutual information $I(A:B|C)$ helps to investigate the tripartite entanglement in Haar random state $\ket{\psi_{ABC}}$.
    
    Using Page's formula, the average conditional mutual information of a Haar random state $\ket{\psi_{ABC}}$ is
    \begin{widetext}
        \begin{equation}
            \bar{I}(A:B|C) = \left\{\begin{array}{lc}
            \frac{D_A D_B}{2 D_C} (1 + O(D^{-\alpha})), & n_C > 1/2 \\ 
            2 \min\{N_A,N_B\} \log 2  + O(D^{-\alpha}) , & \max\{n_A,n_B\} > 1/2 \\
            (N - 2 N_C) \log 2 + O(D^{-\alpha}), & \max\{n_A,n_B,n_C\} < 1/2
            \end{array}\right.
        \end{equation}     
    \end{widetext}
    Applying Levy's lemma~\cite{ledoux2001concentration}, we establish the typicality of the conditional mutual information.
    \begin{theorem} \label{theorem: mutual_information}
        Let $\ket{\psi}_{ABC}$ be a Haar random state on $\mathcal{H}_A \otimes \mathcal{H}_B \otimes \mathcal{H}_C$, where $\dim \mathcal{H}_{X} = D_X = d^{N_X}$, for $X = A, B, C$. 
        In the thermodynamic limit $N\rightarrow\infty$, the conditional mutual information $I(A:B|C)$ converges almost surely to its mean value $\bar{I}(A:B|C)$, i.e., for any given $0<\eta < 1$
        \begin{equation} \label{eq: convergence_I}
            \mathbb{P}(\{|I - \bar{I}| \leq \theta_{\epsilon}\eta\} \ f.e.) = 1,
        \end{equation}
        where for any $\epsilon > 0$, with $n_C = {N_C}/{N}$ 
        \begin{equation}
            \theta_{\epsilon} = \left\{\begin{array}{lc}
                \mathrm{const.}, & n_C \leq \frac{1}{2}, \\
                \bar{I}, & \frac{1}{2} < n_C < \frac{3}{4} - (1+\epsilon) \frac{\log N}{2N\log2} \\
                \frac{2 \log D}{\sqrt{D^{1-\epsilon}}}, & n_C > \frac{3}{4} - (1+\epsilon) \frac{\log N}{2N\log2}.
            \end{array}\right.
        \end{equation}
        Moreover, in the region
        \begin{equation}
            n_C < \frac{3}{4} - (1+\epsilon) \frac{\log N}{2N\log2},
        \end{equation}
        the conditional mutual information almost surely is nonzero with only finite exceptions
        \begin{equation} \label{eq: nonzero_I}
            \mathbb{P}[\{I(A:B|C) > 0\}\ f.e.] = 1.
        \end{equation}
    \end{theorem}\noindent
    The proof is provided in Appendix~\ref{app: mutual_information}.
    In the thermodynamic limit, for $n_{C} >1/2$, the average conditional mutual information is positive but exponentially small.
    Thus, Eq.~(\ref{eq: convergence_I}) implies that $I(A:B|C) \rightarrow 0$ almost surely.
    Although it is possible that $\mathbb{P}[I(A:B|C) = 0] > 0$, for $1/2 < n_{C} < 3/4$, Eq.~(\ref{eq: nonzero_I}) implies that the conditional mutual information is almost surely positive yet sufficient small, i.e., $I(A:B|C) \rightarrow 0^+$.
    We refer to this regime as weakly nonzero.

    Here, we provide an intuitive interpretation of the weakly nonzero conditional mutual information, which is central to the proof.
    For the Lipschitz function $I = I(A:B|C)$, Levy's lemma gives,
    \begin{equation}
        \mathbb{P}(|I - m_I|>\epsilon) \leq 2 \exp\left[-\frac{D \epsilon^2}{2 \log^2 D}\right].
    \end{equation}
    From proportion~1.9 of Ref.~\cite{ledoux2001concentration}, the variance satisfies $\mathrm{Var} I \leq (2 \log^2 D)/{D}$, and the relative standard deviation approaches zero in the thermodynamic limit
    \begin{equation}
        {\Delta I}/{\bar{I}} \leq 2\log D\sqrt{\frac{2D_C}{D_A^3D_B^3}} \leq 2\sqrt{2} N^{-\epsilon} \log 2 \rightarrow 0,
    \end{equation}
    provided $n_C < \frac{3}{4} - (1+\epsilon) \frac{\log N}{2N\log2}$.
    Therefore, $I$ concentrates around its average $\bar{I}$ within the interval $[\bar{I} - C \Delta I, \bar{I} + C \Delta I]$, where $C$ is a finite constant.
    Since ${\Delta I}/{\bar{I}} \rightarrow 0$, typically $\bar{I} - C \Delta I>0$ for sufficiently large $N$, and thus $I$ is typically exponentially small, i.e., weakly nonzero.

    Note that the conditional mutual information $I(A:B|C)$ is related to the quantum Markov recovery problem.
    If $I(A:B|C) = 0$, the state 
    \begin{equation}
        \rho_{ABC} = \sum_i p_i \rho_{AB_1} \otimes \rho_{B_2 C}
    \end{equation}
    forms a quantum Markov chain $A\rightarrow B \rightarrow C$, meaning it can be exactly recovered from the marginal state $\rho_{AB}$ by Petz recovery map $\mathcal{R}_{B \rightarrow BC}^0$
    \begin{equation}
        \rho_{ABC} = \mathcal{R}_{B \rightarrow BC}^0(\rho_{AB}).
    \end{equation}
    For a weakly nonzero conditional mutual information $I(A:B|C) \rightarrow 0^+$, the state $\rho_{ABC}$ can be approximately recovered from $\rho_{AB}$ using a (rotated) Petz recovery map $\mathcal{R}_{B \rightarrow BC}$ with high fidelity or small trace distance~\cite{fawzi2015quantum,PhysRevLett.115.050501,doi:10.1098/rspa.2015.0623}, Eq.~(\ref{eq: markov_recovery})
    \begin{equation}
        \rho_{ABC} \approx \mathcal{R}_{B \rightarrow BC}(\rho_{AB}),
    \end{equation}
    forming an approximate quantum Markov chain.
    Such chain can differ significantly from exact quantum Markov chain~\cite{ibinson2008robustness,christandl2012entanglement,erker2015not} and are well approximated by thermal states of short-range Hamiltonians~\cite{kato2019quantum}.
    Since states forming quantum Markov chain, Eq.~(\ref{eq: Markov_chain}), do not have bound entangled marginal state, Theorem~\ref{theorem: mutual_information} implies that in region $1/2<n_C<n_{\mathrm{SEP}}$, states with PPT marginal states typically form approximate quantum Markov chain that are far from being exact.
    This highlights the non-trivial nature of the weakly nonzero regime.
    In the following section, we establish the existence of a weakly nonzero Markov gap, which relates to partial quantum Markov recovery problem and shares similar properties and interpretations.

    In the region $n_{\mathrm{SEP}} < n_C < 3/4$, the marginal state $\rho_{AB}$ is typically separable, $\mathbb{P}[\mathrm{SEP}_{AB}] \rightarrow 1$, and Theorem~\ref{theorem: mutual_information} implies that the conditional mutual information is typically weakly nonzero, Eq.~(\ref{eq: nonzero_I}).
    From Eq.~(\ref{eq: sep_g}),
    \begin{equation}
        \{I(A:B|C)>0\} = \neg\mathrm{SEP}_{AB} \cup \{g(A:B) > 0\} ,
    \end{equation}
    so that 
    \begin{equation}
        \{g(A:B) > 0\} \supset \{I(A:B|C)>0\} \cap \mathrm{SEP}_{AB}.
    \end{equation}
    Therefore, it follows that
    \begin{align}
        \mathbb{P}[g(A:B)>0] & \geq \mathbb{P}[\{I(A:B|C)>0\} \cap \mathrm{SEP}_{AB}] \nonumber\\
        & \rightarrow \mathbb{P}[\mathrm{SEP}_{AB}] \rightarrow 1,
    \end{align}
    which confirms the existence of tripartite entanglement in the region $n_{\mathrm{SEP}} < n_C < 3/4$.

\section{Weakly Nonzero Markov Gap and Entanglement Undistillability} \label{sec: Markov_gap}

    In this section, we investigate the Markov gap of Haar random state.
    The average reflected entropy for the tripartite Haar ensemble is given by (following the case of the single random tensor in Ref.~\cite{akers2022reflected})
    \begin{widetext}
        \begin{equation}
            \bar{S}_R(A:B) = - p_0 \log p_0 - p_1 \log p_1 + p_1 \left(\log D_A^2 - \frac{D_A^2}{2 D_B^2}\right) + O(D^{-2}),
        \end{equation}    
        where $p_0 = 1 - {D_{AB}}/{4D_C}$ for $D_{AB} \ll D_C$, $p_0 = {D_C}/{D_{AB}}$ for $D_{AB} \gg D_C$, and $p_0 + p_1 = 1$.
        Assuming $D_A < D_B$, the leading order of average Markov gap is [see Fig.~\ref{fig: figure3}(b) and~(c) for diagrams]
        \begin{equation}
            \bar{h}(A:B) = \left\{\begin{array}{lc}
            \frac{D_{AB}}{4D_C}\left[\left(N- 2N_B\right)\log 2 + O(1)\right], & n_C > 1/2 \\
            \frac{D_{AC}}{2D_B} \left(1+ O(D_A^{-2})\right), & n_B > 1/2 \\
            \left(N- 2N_B\right)\log2 + O(D^{-\alpha}), & \max\{n_B,n_C\} < 1/2 
        \end{array}\right.
    \end{equation}    
    \end{widetext}

    The average Markov gap is significantly larger than zero in the region $n_{\max}\leq1/2$, a regime we refer to as strongly nonzero. 
    In contrast, it approaches zero in the region $n_{\max}>1/2$.
    There is no doubt that in the strongly nonzero regime, the Markov gap is nonzero, $\mathbb{P}(h>0) > 0$, implying the existence of non-SOTS-type tripartite entanglement.
    To summarize our findings so far\\
    \textit{(a)} for $n_{\mathrm{SEP}}<n_{\max}<3/4$, $\mathbb{P}(g>0) \rightarrow 1$;\\
    \textit{(b)} for $n_{\mathrm{PPT}} = 1/2 < n_{\max} < n_{\mathrm{SEP}}$, $\mathbb{P}(h>0) \rightarrow 1$;\\
    \textit{(c)} for $n_{\max}\leq1/2$, $\mathbb{P}(h>0) > 0$.

    Applying Levy's Lemma~\cite{ledoux2001concentration} again, we establish the typicality of the Markov gap.
    \begin{theorem} \label{theorem: Markov_gap} 
        Let $\ket{\psi}_{ABC}$ be a Haar random state on $\mathcal{H}_A \otimes \mathcal{H}_B \otimes \mathcal{H}_C$, where $\dim \mathcal{H}_{X} = D_X = d^{N_X}$, for $X = A, B, C$. 
        In the thermodynamic limit $N\rightarrow\infty$, the Markov gap $h$ converges almost surely to its mean value $\bar{h}$, i.e., for any given $0<\eta < 1$

        \begin{equation} \label{eq: convergence_h}
            \mathbb{P}(\{|h - \bar{h}| \leq \theta_{\epsilon}\eta\} \ f.e.) = 1,
        \end{equation}
        where for any $\epsilon > 0$, with the proportion $n_C = \frac{N_C}{N}$  of subsystem $C$
        \begin{equation}
            \theta_{\epsilon} = \left\{\begin{array}{lc}
                \mathrm{const.}, & n_C \leq \frac{1}{2}, \\
                \bar{h}, & \frac{1}{2} < n_C < \frac{3}{4} - \frac{\epsilon\log N}{2N\log2} \\
                \frac{2 \log D}{\sqrt{D^{1-\epsilon}}}, & n_C > \frac{3}{4} - \frac{\epsilon\log N}{2N\log2}.
            \end{array}\right.
        \end{equation}
        Moreover, in the region
        \begin{equation}
            n_{\max} < \frac{3}{4} - \frac{\epsilon\log N}{2N\log2},
        \end{equation}
        the Markov gap is almost surely nonzero with only finite exceptions
        \begin{equation}
            \mathbb{P}[\{h > 0\}\ f.e.] = 1.
        \end{equation}
    \end{theorem}\noindent
    The proof is provided in Appendix~\ref{app: concentration}. 
    This result implies that the Markov gap is weakly nonzero $h \rightarrow 0^+$ for $1/2<n_{\max}<3/4$, and approaches zero $h\rightarrow 0$ (not necessarily nonzero) for $n_{\max} > 3/4$.
    When $n_{\max} <1/2$, Markov gap is strongly nonzero $h \rightarrow \infty$
    \begin{equation}
        \left\{\begin{array}{lc}
        \mathbb{P}(\{h\rightarrow \infty\} \ f.e.) = 1, &  n_{\max}< 1/2, \\
        \mathbb{P}(\{h\rightarrow0^{+}\} \ f.e.) = 1, & 1/2 < n_{\max} <3/4,  \\
        \mathbb{P}(\{h\rightarrow0\} \ f.e.) = 1, &  n_{\max} > 3/4.
        \end{array}\right. 
    \end{equation} 
    
    Similar to Theorem~\ref{theorem: mutual_information}, we provide a brief interpretation of the weakly nonzero Markov gap $h = h(A:B)$.
    The concentration inequality for the Markov gap is
    \begin{equation}
        \mathbb{P}(|h - m_h|>\epsilon) \leq 2 \exp\left[-\frac{D \epsilon^2}{2 (1+2n_A)^2\log^2 D}\right].
    \end{equation}
    The relative standard deviation approaches zero in the thermodynamic limit
    \begin{align}
        {\Delta h}/{\bar{h}} & \leq \frac{4(1+2n_A)}{(1-2n_B)}\sqrt{\frac{2D_C}{D_A^3D_B^3}} \nonumber \\ 
        & \leq \frac{4\sqrt{2}(1+2n_A)}{(1-2n_B)} N^{-\epsilon}  \rightarrow 0,
    \end{align}
    provided $n_C < \frac{3}{4} - \frac{\epsilon\log N}{2N\log2}$.
    Therefore, $h$ concentrates around its average $\bar{h}$ with a sufficiently small dispersion, making it weakly nonzero.

    In the region $n_{\max}<3/4$, the probability of any event $A$ satisfies
    \begin{align}
        \mathbb{P}(A) \rightarrow \mathbb{P}(A \cap \{h > 0\}).
    \end{align}
    Consequently, any behavior of the Haar random state in this region originates from the behavior of states with a nonzero Markov gap.
    In particular, the PPT threshold $n_{\mathrm{PPT}} = 1/2$ corresponds to the transition from a strongly nonzero to a weakly nonzero Markov gap, and the separability threshold $n_{\mathrm{SEP}} = (1 + n_{\min})/2$ emerges in the weakly nonzero regime of Markov gap.

    It also follows that with only finite exceptions in the thermodynamic limit, a state with a weakly nonzero Markov gap almost surely has an undistillable (i.e., bound entangled or separable) marginal state 
    \begin{equation}
        \mathbb{P}(\mathrm{BND}_{AB} \cup \mathrm{SEP}_{AB} | \{h\rightarrow0^{+}\} \ f.e.) = 1,
    \end{equation}
    and a state with a strongly nonzero Markov gap almost surely has NPT marginal states
    \begin{equation}
        \mathbb{P}(\mathrm{NPT}_{AB} | \{h \gg 0\} \ f.e.) = 1.
    \end{equation}
    The marginal undistillability of a state $\ket{\psi}_{ABC}$ with a weakly nonzero Markov gap implies that it typically cannot be reduced to a tensor product of states with zero Markov gap and other states with (likely strongly) nonzero Markov gap, since the marginal state $\rho_{AB}$ of a SOTS is not bound entangled.
    As mentioned earlier, the Markov gap is related to the partial quantum Markov recovery problem.
    Thus, a state with a weakly nonzero Markov gap forms a (partial) approximate quantum Markov chain.
    The marginal undistillability of weakly nonzero Markov gap reveals a possible structure of the PPT bound entangled states, which is left for further investigation.

\section{Holographic Interpretation} \label{sec: holography}

    In holographic duality, both the EoP $E_p$ and the reflected entropy $S_R$ are proposed to be dual to the entanglement wedge cross section~\cite{umemoto2018entanglement,hayden2021markov}. 
    It is conjectured that $S_R = 2 E_p$ for holographic state.
    At leading order, a nonzero Markov gap is dual to a discontinuous boundary of the KRT surfaces $\gamma_{A,\mathrm{KRT}}$, which is suggested to correspond to tripartite entanglement~\cite{akers2020entanglement,hayden2021markov}.
    For a holographic state with a connected entanglement wedge $W(A:B)$, each discontinuous boundary of the KRT surfaces $\gamma_{A,\mathrm{KRT}}$ and $\gamma_{B,\mathrm{KRT}}$ contributes a geometric term $\geq \log2/G_N$ to the Markov gap at leading order~\cite{hayden2021markov}
    If the entanglement wedge is disconnected, the KRT surface has no discontinuous boundary, thus the leading geometric contribution to the Markov gap vanishes.

    Although tripartite entanglement is expected to support the discontinuous boundaries of KRT surfaces $\gamma_{A,\mathrm{KRT}}$, the GHZ state is excluded from the classical holographic states~\cite{hayden2021markov}, due to the violation of the monogamy of mutual information~\cite{PhysRevD.87.046003}, which requires that for any three regions $A, B$ and $C$
    \begin{equation}
        I(A:BC) \geq I(A:B) + I(A:C).
    \end{equation}
    More generally, the tripartite entanglement of a SOTS, which has a zero Markov gap, should also be excluded from the classical holographic states.
    \begin{theorem} \label{theorem: monogamy}
        For the sum of triangle state, $h = 0$, 
        \begin{align}
            I(A:BC_1) & = I(A:B) + I(A:C_1) - g \nonumber\\
            &\leq I(A:B) + I(A:C_1).
        \end{align}
    \end{theorem} \noindent
    Obviously, it contradicts the monogamy of mutual information unless $g=0$, where no tripartite entanglement exists.
    Therefore, $h$ and $g$ likely measure the same type of tripartite entanglement contributing to geometric entropies of a holographic state, supporting the conjecture $S_R = 2 E_p$~\cite{dutta2021canonical}.
    Although SOTS-type tripartite entanglement cannot exist in a classical holographic state, the presence of quantum matter in the bulk geometry generally allows for violations of the monogamy of mutual information.
    Thus, SOTS may present in the quantum corrections of the bulk term $S_{\mathrm{bulk}}$ of entropy perturbatively in semiclassical limit.
    \begin{proof}[Proof of Theorem~\ref{theorem: monogamy}]
        Consider the mutual information $I(A:BC_1)$ for a SOTS
        \begin{equation}
            \ket{\psi_{ABC}} = \sum_l \sqrt{p_l} \ket{\psi_{A_1 B_2}^l} \ket{\psi_{B_1 C_2}^l} \ket{\psi_{C_1 A_2}^l}.
        \end{equation}
        As calculated in Proposition~\ref{proposition: SOTS_g}, we have 
        \begin{align}
            S(A) & = S_{A:B} + S_{A:C} + g, \\
            S(C_1) & = S_{A:C} + g, \\
            S(AC_1) & = S_{A:B} + g, \\
            S(BC1) & = S_{A:C} + S_{B:C} + S_{A:B} + g, \\
            S(ABC_1) & = S_{B:C} + g,
        \end{align}
        where $g = H(p_l)$.
        The mutual informations are then
        \begin{align}
            I(A:BC_1) & = S(A) + S(BC1) - S(ABC1) \nonumber \\
            & = 2 S_{A:B} + 2 S_{A:C} + g, \\
            I(A:B) & = S(A) + S(B) - S(AB) \nonumber \\
            &  = 2 S_{A:B} + g , \\
            I(A:C_1) & = S(A) + S(C_1) - S(AC_1) \nonumber\\
            & = 2 S_{A:C} + g.
        \end{align}
        It follows that
        \begin{align}
            I(A:BC_1) & = I(A:B) + I(A:C_1) - g \nonumber\\
            &\leq I(A:B) + I(A:C_1),
        \end{align}
        with $g \geq 0$.
    \end{proof}

    In the AdS/CFT correspondence, the tripartite Haar random state models a three-boundary wormhole~\cite{balasubramanian2014multiboundary,hayden2016holographic,akers2022reflected}, as depicted in Fig.~\ref{fig: figure4}.
    The areas of the wormhole mouths $\mathcal{A}_{A,B,C}$ are related to the subsystem dimensions
    \begin{equation}
        N_{A,B,C} \log2 \equiv \log D_{A,B,C} = \frac{\mathcal{A}_{A,B,C}}{4G_N},
    \end{equation}
    where the perturbative constant $G_N$ is dual to the number of qubits $N \propto 1/G_N$.
    The reflected entropy of Haar random states exhibits an entanglement wedge phase transition from connected to disconnected at $n_{\max} = 1/2$, coinciding with the PPT threshold for Haar random states.
    This phase transition is discontinuous in the leading geometric contribution, as the Markov gap receives a contribution $\geq \log2/G_N$ from each discontinuous boundary of the KRT surface.
    This discontinuity is smoothed out by the non-perturbative effects~\cite{akers2022reflected}.

    In the previous section, we identified two distinct regimes of nonzero Markov gap, strongly and weakly nonzero.
    The strongly nonzero Markov gap, scaling as $\propto N \propto 1/G_N$, corresponds to the tripartite entanglement dual to the discontinuous boundary of the KRT surface at leading order.
    The weakly nonzero Markov gap, scaling as $\propto O(D^{-\alpha}) \propto O(e^{-1/G_N})$ corresponds to non-perturbative effects.
    Consequently, the PPT bound entanglement of the boundary state emerges from the non-perturbative contributions in the bulk geometry.
    In brief, we supposed that tripartite entanglement with a strongly nonzero Markov gap is dual to geometric contribution, tripartite entanglement with a vanishing Markov gap is dual to bulk contribution at perturbative order, and tripartite entanglement with a weakly nonzero Markov gap is dual to non-perturbative effects.

\section{Conclusion} \label{sec: conclusion}

In Haar random states, the separation between the thresholds for separability and PPT implies the existence of bound entanglement.
Since bound entanglement is not a direct resource in quantum computation and communication tasks, understanding its characteristics is crucial. 
Our work demonstrates that bound entangled states possess a nonzero Markov gap, suggesting a fundamental connection between these two concepts.

We have investigated tripartite entanglement in Haar random states using the Markov gap.
Our results show that  Haar random states almost surely exhibit a nonzero Markov gap when $n_{\max}<3/4$.
In particular, when $1/2<n_{\max}<3/4$, the Markov gap is almost surely weakly nonzero.
Therefore, the transition from bound entangled to separable state can be understood as a consequence of the properties inherent to state with a weakly nonzero Markov gap.

The Markov gap is a quantity that was recently introduced in the context of AdS/CFT research.
We have proven that tripartite entangled states with a vanishing Markov gap generally violate the monogamy of mutual information.
This indicates that such states are dual to the bulk entanglement entropy at the perturbative level.
Furthermore, since the weakly nonzero Markov gap corresponds to non-perturbative contributions in holographic duality, we conclude that PPT bound entanglement in the boundary emerges from the non-perturbative effects in bulk geometry.
Our findings illuminate potential pathways for further investigation of Markov gap, and enhance the interdisciplinary application of quantum information.

\begin{acknowledgments}
H. F. acknowledges support from the National Natural Science Foundation of China (Grants No.~T2121001, No.~92265207,  No.~92365301), the Innovation Program for Quantum Science and Technology (Grant No.~2021ZD0301800).
S. L. acknowledges support from the Gordon and Betty Moore Foundation under Grant No.~GBMF8690, the National Science Foundation under Grant No.~NSF PHY-1748958, and the Simons Foundation under an award to Xie Chen (Award No.~828078).
\end{acknowledgments}

\appendix

\section{Additional Introductions} \label{app: additional}
In this appendix, we introduce some additional concepts and definitions used in the main text.
\subsection{Adverbs in probability theory} 

In this Appendix, we introduce the precise definitions of some adverbs used in main text.
For more details, see Ref.~\cite{shiryaev2016probability}.
Let $A$ denote event in sample space $\Omega$, the opposite event is $\bar{A}=\Omega\setminus A$.

\begin{definition}
    For a sequence of events $A_n, (n =1, 2, \dots \infty)$, the events occur \textit{in probability} if 
    \begin{equation}
        \mathbb{P}(A_n) \rightarrow 1.
    \end{equation}
\end{definition}\noindent
Typically, a sequence of random variable $\xi_n$ converges to $\xi$ in probability $\xi_n \overset{\mathbb{P}}{\rightarrow} \xi$ if for any $\epsilon>0$,
\begin{equation}
    \mathbb{P}(|\xi_n - \xi|>\epsilon) \rightarrow 0.
\end{equation}

\begin{definition}
    Event $A$ is \textit{almost surely} with respect to probability measure $\mathbb{P}$ if 
    \begin{equation}
        \mathbb{P}(\bar{A}) = 0.
    \end{equation}
\end{definition}

\begin{definition}
    For a sequence of events $A_n, (n =1, 2, \dots \infty)$, the event that infinitely many of events $A_k$ occur is 
    \begin{equation}
        A_n \ i.o. \equiv \lim\sup A_n \equiv \bigcap_{n=1}^{\infty} \bigcup_{k\geq n} A_k,
    \end{equation} 
    where $i.o.$ is the abbreviation of \textit{infinitely often}.
\end{definition}

\begin{definition}
    For a sequence of events $A_n, (n =1, 2, \dots \infty)$, the event that events $A_k$ occur with only finite exceptions is 
    \begin{equation}
        A_n \ f.e. \equiv \lim\inf A_n \equiv \bigcup_{n=1}^{\infty} \bigcap_{k\geq n} A_k,
    \end{equation} 
    where $f.e.$ is the abbreviation of \textit{finite exceptions}.
\end{definition}\noindent
Obviously, $\lim\inf A_n \subseteq \lim\sup A_n$.
Moreover, with De Morgan's laws, we have 
\begin{align}
    \overline{\lim\inf A_n} = \lim\sup \bar{A}_n, \\
    \overline{\lim\sup A_n} = \lim\inf \bar{A}_n. 
\end{align}
In addition, it also has
\begin{align} \label{eq: event_compose}
    \lim \sup (A_n \cup B_n) & = \lim \sup A_n \cup \lim \sup B_n, \\
    \lim \inf (A_n \cap B_n) & = \lim \inf A_n \cap \lim \inf B_n.
\end{align}

For a sequence of events $A_n, (n =1, 2, \dots \infty)$, the events occur almost surely with only finite exceptions means that
\begin{equation}
    \mathbb{P}(\lim\inf A_n) = 1.
\end{equation}
In particular, a sequence of random variable $\xi_n$ converge to $\xi$ means for any $\epsilon$, 
\begin{equation}
    |\xi_n - \xi| < \epsilon 
\end{equation}
with only finite exceptions.
Therefore, the sequence of random variable $\xi_n$ converge to $\xi$ almost surely $\xi_n \overset{a.s.}{\rightarrow} \xi$ if for any $\epsilon>0$,
\begin{equation}
    \mathbb{P}(|\xi_n - \xi|>\epsilon \ i.o.) = 0.
\end{equation}
It is also called as the convergence \emph{with probability one}.

\begin{lemma} \label{lemma: f.e.}
    The sequence of event $A_n$ almost surely occurs with only finite exceptions
    \begin{equation}
        \mathbb{P}(\lim\inf A_n) = 1,
    \end{equation}
    if and only if events $\sup \bar{A}_n \equiv \bigcup_{k\geq n} \bar{A}_k$ is not occur in probability 
    \begin{equation}
        \mathbb{P}(\sup \bar{A}_n)  \rightarrow 0.
    \end{equation}
    Moreover, a sufficient condition is that the series converges
    \begin{equation}
        \sum_{n=1}^{\infty} \mathbb{P}(\bar{A}_n) < \infty.
    \end{equation}
\end{lemma}    
\begin{proof}
    The condition that events $A_n$ almost surely occurs with only finite exceptions
    \begin{equation}
        \mathbb{P}(\lim\inf A_n) = 1 
    \end{equation}
    is equivalent to 
    \begin{equation}
        \mathbb{P}(\lim\sup \bar{A}_n) = 0.
    \end{equation}
    Since 
    \begin{align}
        \mathbb{P}(\lim\sup \bar{A}_n) & \equiv \mathbb{P}\left(\bigcap_{n=1}^{\infty} \bigcup_{k\geq n} \bar{A}_k\right) \nonumber \\
        & = \lim_{n\rightarrow\infty} \mathbb{P}\left(\bigcup_{k\geq n} \bar{A}_k\right) \nonumber \\
        & \equiv \lim_{n\rightarrow\infty} \mathbb{P}(\sup \bar{A}_n),
    \end{align}
    it is equivalent to 
    \begin{equation}
        \mathbb{P}(\sup \bar{A}_n)  \rightarrow 0.
    \end{equation}

    Moreover, when the series converges
    \begin{equation}
        S = \sum_{n=1}^{\infty} \mathbb{P}(\bar{A}_n) < \infty,
    \end{equation}
    the partial sum $S_N = \sum_{n=1}^{N} \mathbb{P}(\bar{A}_n)$ converges to $S$
    \begin{equation}
        \left|\sum_{k\geq n} \mathbb{P}(\bar{A}_k)\right| = |S-S_n| \rightarrow 0.
    \end{equation}
    Since 
    \begin{equation}
        \mathbb{P}(\sup \bar{A}_n) \equiv \mathbb{P}\left(\bigcup_{k\geq n} \bar{A}_k\right) \leq \sum_{k\geq n} \mathbb{P}(\bar{A}_k),
    \end{equation}
    it follows that 
    \begin{equation}
        \mathbb{P}(\lim\inf A_n) = 1.
    \end{equation}
\end{proof}

\subsection{Multipartite entanglement} 

    A $k$-partite state $\rho$ is entanglement if it is not separable, i.e. it cannot be decomposed as
    \begin{equation}
        \rho = \sum_i p_i \rho_1^i \otimes \dots\otimes \rho_k^i.
    \end{equation}
    The multipartite entanglement is more complicated than bipartite entanglement.
    For example, there is a special type of multipartite entangled state called semiseparable state~\cite{PhysRevLett.82.5385}, where the $k$-partite state is separable with respect to any $m$-partition such that $m\leq k-1$, while it is entangled for the $k$-partition.

    Moreover, a bipartite pure state always has the Schmidt decomposition
    \begin{equation}
        \ket{\psi}_{AB} = \sum_i \sqrt{p_i} \ket{i_A} \ket{i_B},
    \end{equation}
    where the Schmidt coefficients $p_i$ are the eigenvalues of its marginal states
    \begin{equation}
        \rho_A \equiv \mathrm{Tr}_B\ket{\psi}\bra{\psi}_{AB}, \quad \rho_B \equiv \mathrm{Tr}_A\ket{\psi}\bra{\psi}_{AB},
    \end{equation} 
    and $\ket{i_A}$ and $\ket{i_B}$ are orthonormal basis of system $A$ and $B$, respectively.
    The Schmidt coefficients are invariant under local unitary operations, and uniquely determine the bipartite entanglement.
    However, a multipartite pure state does not always have the general Schmidt decomposition $\sum_i c_i \ket{i_1} \otimes \dots \otimes \ket{i_k}$~\cite{PERES199516,PhysRevA.59.3336}, such as the W state~\cite{PhysRevA.62.062314}
    \begin{equation}
        \ket{W} = (\ket{100} + \ket{010} + \ket{001})/\sqrt{3}.
    \end{equation}

    For the multipartite state admitting Schmidt decomposition, GHZ state 
    \begin{equation}
        \ket{GHZ} = \sum_{i=1}^d \ket{i_1} \otimes \dots \otimes \ket{i_k}/\sqrt{d}
    \end{equation}
    seems to be maximally entangled in the sense of entanglement distillation~\cite{GISIN19981}.
    However, for a general multipartite state, a set of inequivalent maximally entangled states is required~\cite{PhysRevA.62.062314,PhysRevLett.111.110502}.

\subsection{Haar measure and Haar induced measure} 

    The Haar measure can be characterized by its $n$-fold moments
    \begin{equation}
        \Phi_{\mathrm{Haar}}^{(n)} (\cdot) \equiv \mathbb{E}_{\hat{U}}[\hat{U}^{\otimes n} (\cdot) \hat{U}^{\dagger \otimes n}] 
    \end{equation}
    which can be evaluated by the Weingarten calculus~\cite{10.1155/S107379280320917X,collins2006integration}
    \begin{equation} \label{eq: Haar}
        \Phi_{\mathrm{Haar}}{(k)} (\cdot)= \sum_{\sigma, \tau \in S_n} \mathrm{Wg}(\tau\sigma^{-1}) \hat{\Pi}_{\tau} \mathrm{Tr}(\hat{\Pi}_{\sigma} \cdot) ,
    \end{equation}
    where $\mathrm{Wg}(\tau \sigma^{-1})$ is the Weingarten function~\cite{collins2006integration}, and $\hat{\Pi}_{\tau}$ is the representation of permutation $\tau \in S_n$ on $\mathcal{H}^{\otimes n}$.
    In the thermodynamic limit, the Weingarten function is
    \begin{equation}
        \mathrm{Wg}(\tau\sigma^{-1}) = D^{-d(\sigma,\tau)} \left[\mathrm{Mob}(\tau\sigma^{-1}) + O\left(D^{-2}\right)\right],
    \end{equation}
    where $d(\sigma,\tau) = n - \#(\sigma^{-1}\tau)$ defines an integer-valued distance on the permutation group $S_{n}$, $\#(\tau)$ is the number of cycles in the permutation $\tau$, and $\mathrm{Mob}(\sigma)$ is the M\"obius function of the permutation $\sigma$ ~\cite{collins2016random,nica2006lectures}.

    For a random pure state $\ket{\psi}$ distributed in Haar measure, if the system is divided into system $A$ and $B$, whose Hilbert space is $\mathcal{H}_A \otimes \mathcal{H}_B$ with dimensions $D_A$ and $D_B$ respectively, 
    \begin{equation}
        \ket{\psi} = \sum_{i_A, j_B} X_{i_A, i_B} \ket{i_A} \ket{j_B},
    \end{equation}
    where $\ket{i_A} \ket{j_B}$ is the orthogonal basis of the Hilbert space $\mathcal{H}_A \otimes \mathcal{H}_B$.
    Its marginal state of system $A$ is
    \begin{equation}
        \rho_{A} = \mathrm{Tr}_B(\ket{\psi}\bra{\psi}) = X X^{\dagger}.
    \end{equation}
    This state $\rho_{AB} \sim \nu_{D_{AB}, D_C}$ is distributed in an induced measure $\nu_{D_{A}, D_B}$ of the Haar measure, which is called the Haar random induced state. 
    For the case $D_B \geq D_A$, it is shown that~\cite{zyczkowski2001induced}
    \begin{equation}
        \frac{\mathrm{d} \mu_{D_A,D_B}}{\mathrm{d} m} = \frac{(\det \rho)^{D_B - D_A}}{Z_{D_A, D_B}},
    \end{equation}
    where $m$ is the Lebesgue measure of the set $\mathcal{D}(\mathcal{H}_A)$ of all quantum states on $\mathcal{H}_A$, $Z_{D_A, D_B}$ is normalization constant.
    This state can also be described by the Wishart ensemble in thermodynamic limit~\cite{forrester2010log}, where the matrix elements $X_{ij}$ satisfy the Gaussian distribution.    
    The Haar induced measure $\nu_{D_{A}, D_B}$ is the product of independent distributions of its eigenstates and eigenvalues~\cite{collins2016random,zyczkowski2001induced}.
    To be more precise, it is 
    \begin{equation}
        \rho_{A} = \hat{U} \hat{D} \hat{U}^{\dagger},
    \end{equation}
    where $\hat{U}$ is random unitary under Haar measure on $\mathcal{H}_A$, and $\hat{D}$ is a diagonal matrix, whose empirical spectrum measure 
    \begin{equation}
        \mu_{D} = \frac{1}{D_{A}} \sum_{i = 1}^{D_{A}} \delta_{\lambda_i},
    \end{equation}
    almost surely converges weakly to M\v{a}rcenko-Pastur (MP) law $\mu$ in thermodynamic limit~\cite{collins2016random,marchenko1967distribution,nadal2011statistical,PhysRevA.85.062331}.
    \begin{align}
        \mathrm{d}\mu & = \frac{\sqrt{4 c \tau^2 - (\lambda - c \tau - \tau)^2}}{2 \pi \tau \lambda} I_{(\lambda_-, \lambda_+ )}(\lambda) \mathrm{d}\lambda \nonumber\\
        &~~~~ + \max(1 - c, 0) \delta(\lambda),
    \end{align}
    where $\tau = \frac{1}{D_B}$, $c = \frac{D_B}{D_A}$, and $\lambda_{\pm} =  \tau(1\pm\sqrt{c})^2$.

\section{Proofs for Existence of Tripartite Entanglement in Haar Ensemble} \label{app: mutual_information}

To prove Theorem~\ref{theorem: mutual_information}, we use the Levy's Lemma~\cite{ledoux2001concentration}.
\begin{lemma}[Levy's Lemma]
    If $f: \mathbb{S}^{D-1} \rightarrow \mathbb{R}$ is an $L$-Lipschitz function, then for every $\epsilon > 0$,
    \begin{equation}
        \mathbb{P}(|f - m_f|>\epsilon) \leq 2 \alpha(\epsilon/L),
    \end{equation}
    where $\alpha(r) = e^{-2 D r^2}$, and $m_f$ is the median of $f$ for $\mathbb{P}$, such that 
    \begin{equation}
        \mathbb{P}(f\leq m_f) \geq 1/2, \quad\textrm{and}\quad \mathbb{P}(f\geq m_f) \geq 1/2.
    \end{equation}
\end{lemma} \noindent
Here, the median $m_f$ can be substituted by other central value like the mean value without change the concentration behaviors of measure.
We first show that the mutual information is a Lipschitz function.
\begin{lemma}
    For state $\ket{\psi} \in \mathcal{H}_A \otimes \mathcal{H}_B \otimes \mathcal{H}_C$, the entropy $S_A$ is a $2 \log D_A$-Lipschitz function, and consequently the mutual information $I(A:B)$ is a $2 \log D$-Lipschitz function.
\end{lemma}
\begin{proof}
    To show a function is $L$-Lipschitz is equivalent to show that its gradient is bounded $|\nabla f| \leq L$.
    The differential of entropy is 
    \begin{align}
        \mathrm{d} S_A & = - \mathrm{Tr}_A[\mathrm{d}\rho_A \log \rho_A] \nonumber\\
        & = -\mathrm{Tr}[\mathrm{d}(\ket{\psi}\bra{\psi}) \log \rho_A \otimes I_{\bar{A}}] \nonumber \\
        & = -2 \Re \Braket{\psi|\log \rho_A \otimes I_{\bar{A}}|\mathrm{d}\psi}.
    \end{align}
    Therefore, the gradient of entropy about the state $\ket{\psi}$ is 
    \begin{equation}
        \nabla_{\psi} S_A = - 2 (\log \rho_A \otimes I_{\bar{A}})\ket{\psi},
    \end{equation}
    whose norm is 
    \begin{align}
        \Vert \nabla_{\psi} S_A\Vert_2 & = 2 \sqrt{\Braket{\psi|\log^2 \rho_A \otimes I_{\bar{A}}|\psi}} \nonumber\\
        & = 2 \sqrt{\mathrm{Tr}_A[\rho_A \log^2 \rho_A]} \leq 2 \log D_A.
    \end{align}
    Since the mutual information $I(A:B) = S_A + S_B - S_C$, its gradient  
    \begin{align}
        \Vert \nabla_{\psi} I(A:B)\Vert_2 & \leq \Vert \nabla_{\psi} S_A\Vert_2 + \Vert \nabla_{\psi} S_B\Vert_2 + \Vert \nabla_{\psi} S_C\Vert_2 \nonumber \\
        & = 2 \log D.
    \end{align}
\end{proof}

\begin{proof}[Proof of Theorem~\ref{theorem: mutual_information}]
    With Levy's Lemma, we have 
    \begin{equation}
        \mathbb{P}(|I - m_I|>\eta) \leq 2 \exp\left[-\frac{D \eta^2}{2 \log^2 D}\right],
    \end{equation}
    where $I \equiv I(A:B|C)$ is the conditional mutual information.
    If $n_C \leq \frac{1}{2}$, the average conditional mutual information
    \begin{equation}
        \bar{I}(A:B|C) \propto N + O(D^{-\alpha}) \rightarrow \infty,
    \end{equation}
    thus there is no difficulty to prove the theorem with Levy's Lemma and Lemma~\ref{lemma: f.e.}.
    Then, we focus on the region $n_C \geq 1/2$.
    By Proposition~1.9 in~\cite{ledoux2001concentration}, we have 
    \begin{equation}
        |\bar{I} - m_I| \leq \log D\sqrt{\frac{2\pi}{D}}.
    \end{equation}
    In Haar ensemble, for $n_C >1/2$, the average mutual information $\bar{I} = \frac{D_A D_B}{2 D_C} (1 + O(D^{-\alpha}))$, thus if $n_C \leq \frac{3}{4} - (1+\epsilon) \frac{\log N}{2N\log 2}$, then
    \begin{align}
        \eta_0 & \equiv |1 - m_I/\bar{I}| \leq 2\log D\sqrt{\frac{2\pi D_C}{D_A^3D_B^3}} (1 + O(D^{-\alpha})) \nonumber \\
        & \leq 2 \sqrt{2\pi} N^{-\epsilon} \log 2 (1 + O(D^{-\alpha})) \rightarrow 0,
    \end{align}
    in the thermodynamic limit $N \rightarrow \infty$.
    This means that $m_I \geq \bar{I}(1-\eta_0) \sim (1-\eta_0)\frac{D_A D_B}{2 D_C}$.
    Since $\eta_0 \rightarrow 0$, for any given $\eta <1$, there is $N_0$ such that for any $N \geq N_0$, $\eta_0<\eta$. 
    \begin{align} \label{eq: concentration_I_bar}
        & \mathbb{P}(|I - \bar{I}| \geq \bar{I}\eta)  \leq \mathbb{P}[|I - m_I| \geq \bar{I}(\eta - \eta_0)] \nonumber\\
        & \leq 2 \exp\left[- \frac{(\eta - \eta_0)^2 D_A^3 D_B^3 }{8 D_C \log^2 D} (1 + O(D^{-\alpha}))\right] \nonumber \\
        & \leq 2 \exp\left[- \frac{(\eta - \eta_0)^2 N^{2\epsilon}}{8\log^2 2} (1 + O(D^{-\alpha}))\right] \rightarrow 0.
    \end{align}

    Since $\frac{(\eta - \eta_0)^2}{8\log^2 2} (1 + O(D^{-\alpha})) \rightarrow \frac{\eta^2}{8\log^2 2} > 0$, for any $\delta>0$, there is a sufficiently large $N_0$, such that $0 < C_{\delta} = \frac{\eta^2(1- \delta)}{8\log^2 2}  < \frac{(\eta - \eta_0)^2}{8\log^2 2} (1 + O(D^{-\alpha}))$.
    By Lemma~\ref{lemma: f.e.}, since 
    \begin{align}
        & \sum_{N=N_0}^{\infty}\mathbb{P}(|I - \bar{I}| \geq \bar{I}\eta)  \leq \sum_{N=N_0}^{\infty} 2 \exp\left[- C_{\delta} N^{2\epsilon}\right] \nonumber \\
        & \leq 2 \int_0^{\infty} \mathrm{d}x e^{-C_{\delta} x^{2\epsilon}} = \frac{\Gamma\left(\frac{1}{2 \epsilon}\right)}{\epsilon C_{\delta}},
    \end{align} 
    it follows that for any $\eta < 1$, with only finite exceptions, $|I_N - \bar{I}| \leq \bar{I}\eta$ 
    \begin{equation} \label{eq: eta_I}
        \mathbb{P}(\{|I - \bar{I}| \leq \bar{I}\eta\} \ f.e.) = 1.
    \end{equation}
    Notice that $\{|I - \bar{I}| \leq \bar{I}\eta\} \subset \{I > \bar{I}(1-\eta) > 0 \}$, we have
    \begin{equation}
        \mathbb{P}[\{I > 0\}\ f.e.] \geq \mathbb{P}(\{|I - \bar{I}| \leq \bar{I}\eta\} \ f.e.) = 1.
    \end{equation}
    Moreover, notice that $I \rightarrow \bar{I}$ if $|I - \bar{I}| \leq \bar{I}\eta$ for any $\eta > 0$, with only finite exceptions, Eq.~(\ref{eq: eta_I}) implies that $I \rightarrow \bar{I}$ almost surely
    \begin{equation}
        \mathbb{P}(I \rightarrow \bar{I}) = 1.
    \end{equation}

    If $n_C \geq 3/4 - (1+\epsilon) \frac{\log N}{2N\log 2}$, $|\bar{I} - m_I|/\bar{I} \rightarrow \infty$, thus for any $\eta < 1$, $\mathbb{P}(|I - \bar{I}| \geq \bar{I}\eta)$ no longer converges exponentially, and $\mathbb{P}[\{I > 0\}\ f.e.] = 1$ is not satisfied.  
    However, the convergence of $I$ to $\{I\}$ can still hold, which will be established in the following.
    For any $\epsilon > 0$ denote 
    \begin{equation}
       \theta_{\epsilon} = \frac{\log D}{\sqrt{D^{1-\epsilon}}} \rightarrow 0.
    \end{equation}
    Now, we denote  
    \begin{equation}
        \eta_0 \equiv |\bar{I} - m_I|/\theta_{\epsilon} \leq \frac{\sqrt{2 \pi}}{D^{\epsilon/2}} \rightarrow 0.
    \end{equation}
    This means that 
    \begin{align} \label{eq: concentration_I_theta}
        & \mathbb{P}[|I - \bar{I}| \geq \theta_{\epsilon}\eta]
        \leq \mathbb{P}[|I - m_I| \geq \theta_{\epsilon}(\eta - \eta_0)] \nonumber\\
        & \leq 2 \exp\left[-\frac{(\eta - \eta_0)^2}{2}D^{\epsilon}\right] \rightarrow 0.
    \end{align}
    Again, since $\frac{(\eta - \eta_0)^2}{2} \rightarrow \frac{\eta^2}{2} > 0$, for any $\delta>0$, there is a sufficiently large $N_0$, such that $0 < C_{\delta} = \frac{\eta^2}{2}(1- \delta)  < \frac{(\eta - \eta_0)^2}{2}$.
    \begin{align}
        & \sum_{N=N_0}^{\infty}\mathbb{P}_N(|I - \bar{I}| \geq \theta_{\epsilon}\eta)  \leq \sum_{N=N_0}^{\infty} 2 \exp\left[-C_{\delta} D^{\epsilon}\right] \nonumber \\
        & \leq 2 \int_0^{\infty} \mathrm{d}x e^{-C_{\delta} 2^{\epsilon x}} = \frac{2}{\epsilon \log 2} \int_{C_{\delta}}^{\infty} \mathrm{d}(\log z) e^{-z} \nonumber \\
        & = \frac{2}{\epsilon \log 2} \left( \int_{C_{\delta}}^{\infty} \log z e^{-z} \mathrm{d}z - \log C_{\delta} e^{-C_{\delta}} \right)\nonumber \\
        & = \frac{2}{\epsilon \log 2} \left[\partial_{\alpha} \left.\left(\int_{C_{\delta}}^{\infty} e^{(\alpha-1)z} \mathrm{d}z\right)\right|_{\alpha = 0} - \log C_{\delta} e^{-C_{\delta}} \right] \nonumber \\
        & = \frac{2}{\epsilon \log 2} \left[\partial_{\alpha} \left.\left(\frac{e^{-(1-\alpha)C_{\delta}}}{1-\alpha} \right)\right|_{\alpha = 0} - \log C_{\delta} e^{-C_{\delta}} \right] \nonumber \\
        & = \frac{2 (C_{\delta}- \log C_{\delta} +1) e^{-C_{\delta}}}{\epsilon \log 2}  < \infty, 
    \end{align}
    we get
    \begin{equation}
        \mathbb{P}(\{|I - \bar{I}| \leq \theta_{\epsilon}\eta\} \ f.e.) = 1.
    \end{equation}
    However, since $\theta_{\epsilon} \geq \bar{I}$ if $n_C > 3/4$, it is possible
    \begin{equation}
        \mathbb{P}[I(A:B|C) = 0] > 0.
    \end{equation}
\end{proof}
\begin{remark}
    The probability inequalities Eqs.~(\ref{eq: concentration_I_bar}) and~(\ref{eq: concentration_I_theta}) are the rate of convergence in probability.
    Given a tolerant probability of failure $\delta$, it can be derived from these inequalities that the minimal size $N_{\delta,\eta}$ of system such that $I$ converges to it mean $\bar{I}$ with given accuracy $\eta$, i.e. $|I - \bar{I}|\leq \eta$.
\end{remark}
\section{Proofs for Weakly Non-Zero Markov Gap} \label{app: concentration}

\begin{proof}[Proof of Theorem~\ref{theorem: Markov_gap}]
    
    Now, assume the subsystem $C$ is the maximal subsystem, $n_{\max} = n_C$, and $n_B \geq n_A$.
    If $n_C \leq \frac{1}{2}$, then
    \begin{align}
        \bar{h}(A:B) = (N - 2N_B)\log2 \rightarrow \infty, \\
        \bar{h}(A:C) = (N - 2N_C)\log2 \rightarrow \infty, \\
        \bar{h}(B:C) = (N - 2N_C)\log2 \rightarrow \infty,
    \end{align}
    in the thermodynamic limit $N \rightarrow \infty$.
    There is no difficult to prove the theorem with Levy's Lemma.

    Then, we focus on the region $n_C \geq \frac{1}{2}$.
    For simplicity, we now denote $h(A:B)$ as $h$.
    Since Markov gap $h$ is a $2(2\log D_A + \log D)$-Lipschitz function, by Levy's Lemma,
    \begin{equation}
        \mathbb{P}(|h - m_h|>\epsilon) \leq 2 \exp\left[-\frac{D \epsilon^2}{2 (1+2n_A)^2N^2\log^2 2}\right].
    \end{equation}
    By Proposition~1.9 in~\cite{ledoux2001concentration}, we have 
    \begin{equation}
        |\bar{h} - m_h| \leq (1+2n_A)N\log 2\sqrt{\frac{2\pi}{D}} \rightarrow 0.
    \end{equation}

    In Haar ensemble, for $n_C >1/2$, the Markov gap $\bar{h}(A:B) = \frac{D_A D_B}{4 D_C} [(1-2n_B)N\log2 + O(1)]$, thus if $n_C \leq \frac{3}{4} - \frac{\epsilon\log N}{2N\log2}$
    \begin{align}
        \eta_0 & \equiv |1 - m_h/\bar{h}| \leq \frac{4(1+2n_A)}{1-2n_B + O(N^{-1})}\sqrt{\frac{2\pi D_C}{D_A^3D_B^3}} \nonumber \\
        & \leq \frac{4(1+2n_A)}{1-2n_B} \sqrt{2\pi} N^{-\epsilon} (1 + O(N^{-1}))\rightarrow 0
    \end{align}
    in the thermodynamic limit $N \rightarrow \infty$.
    This means that 
    \begin{equation}
        m_h \geq \bar{h}(1-\eta_0) \sim (1-\eta_0)\frac{D_A D_B}{4 D_C}(1-2n_B)N\log2.
    \end{equation}
    Since $\eta_0 \rightarrow 0$, for any given $\eta <1$, there is $N_0$ such that for any $N \geq N_0$, $\eta_0<\eta$. 
    \begin{align} \label{eq: concentration_h_bar}
        & \mathbb{P}[|h - \bar{h}| \geq \bar{h}\eta] 
        \leq \mathbb{P}[|h - m_h| \geq \bar{h}(\eta - \eta_0)] \nonumber  \\
        & \leq 2 \exp\left[- \frac{(\eta - \eta_0)^2 (1-2n_B)^2D_A^3 D_B^3 }{32(1+2n_A)^2 D_C}(1 + O(N^{-1}))\right] \nonumber \\
        & \leq 2 \exp\left[- \frac{(\eta - \eta_0)^2 (1-2n_B)^2 N^{2(1+\epsilon)}}{32(1+2n_A)^2}(1 + O(N^{-1}))\right] \nonumber \\
        & \rightarrow 0. 
    \end{align}
    By Lemma~\ref{lemma: f.e.}, similar to proof of Theorem~\ref{theorem: mutual_information}, for arbitrary $\delta > 0$, there is sufficiently large $N_0$, such that $C_{\delta} = \frac{\eta ^2 (1-2n_B)^2 }{32(1+2n_A)^2} (1 -\delta) < \frac{(\eta - \eta_0)^2 (1-2n_B)^2}{32(1+2n_A)^2}(1 + O(N^{-1}))$ for all $N > n_0$.
    We have
    \begin{align}
        &\sum_{N=N_0}^{\infty}\mathbb{P}(|h - \bar{h}| \geq \bar{h}\eta) \leq \sum_{N=N_0}^{\infty} 2 \exp\left[- C_{\delta} N^{2\epsilon}\right] \nonumber \\
        & \leq \sum_{N=1}^{\infty} 2 \exp\left[- C_{\delta} N^{2\epsilon}\right]
        \leq \frac{\Gamma\left(\frac{1}{2 \epsilon}\right)}{\epsilon C_{\delta}} < \infty.
    \end{align} 
    It follows that for every $\eta<1$, $|h - \bar{h}| \leq \bar{h}\eta$ almost surely with only finite exceptions
    \begin{equation} \label{eq: eta_h}
        \mathbb{P}(\{|h - \bar{h}| \leq \bar{h}\eta\}\ f.e.) = 1.
    \end{equation}
    Since $\{|h - \bar{h}| \leq \bar{h}\eta\} \subset \{h \geq \bar{h}(1-\eta) > 0 \}$, it follows $h$ is almost surely larger than zero with only finite exceptions 
    \begin{equation}
        \mathbb{P}(\{h > 0\}\ f.e.) = 1.
    \end{equation}
    Moreover, $h \rightarrow \bar{h}$ if $|h-\bar{h}| < \epsilon$ for any $\epsilon > 0$ with only finite exceptions, Eq.~(\ref{eq: eta_h}) implies that
    \begin{equation}
        \mathbb{P}(h \rightarrow \bar{h}) = 1.
    \end{equation}

    If $n_C \geq 3/4$, $|\bar{h} - m_h|/\bar{h} \rightarrow \infty$, thus for any $\eta < 1$, $\mathbb{P}(|h - \bar{h}| \geq \bar{h}\eta)$ no longer converges exponentially, and $\mathbb{P}[\{h > 0\}\ f.e.] = 1$ is not satisfied.  
    However, the convergence of $I$ to $\{I\}$ can still hold, which will be established in the following.
    For any $\epsilon > 0$ denote 
    \begin{equation}
       \theta_{\epsilon} = \frac{\log D}{\sqrt{D^{1-\epsilon}}} \rightarrow 0,
    \end{equation}
    then 
    \begin{equation}
        \eta_0 \equiv |\bar{h} - m_h|/\theta_{\epsilon} \leq \frac{(1+2n_A)\sqrt{2 \pi}}{D^{\epsilon/2}} \rightarrow 0.
    \end{equation}
    This means that 
    \begin{align} \label{eq: concentration_h_theta}
        & \mathbb{P}[|h - \bar{h}| \geq \theta_{\epsilon}\eta]
        \leq \mathbb{P}[|h - m_h| \geq \theta_{\epsilon}(\eta - \eta_0)] \nonumber\\
        & \leq 2 \exp\left[-\frac{(\eta - \eta_0)^2}{2 (1+2n_A)^2} D^{\epsilon}\right] \rightarrow 0.
    \end{align}
    Again, for arbitrary $\delta > 0$, there is sufficiently large $N_0$, such that $C_{\delta} = \frac{\eta^2 (1 -\delta)}{2 (1+2n_A)^2}  < \frac{(\eta - \eta_0)^2}{2 (1+2n_A)^2}$ for all $N > n_0$.
    We have 
    \begin{align}
        & \sum_{N=N_0}^{\infty}\mathbb{P}_N(|h - \bar{h}| \geq \theta_{\epsilon}\eta)  \leq \sum_{N=N_0}^{\infty} 2 \exp\left[-C_{\delta} D^{\epsilon}\right] \nonumber \\
        & \leq  \frac{2 (C_{\delta}- \log C_{\delta} +1) e^{-C_{\delta}}}{\epsilon \log 2}  < \infty, 
    \end{align}
    By Lemma~\ref{lemma: f.e.}, it follows that $h$ almost surely converges to $\bar{h}$ 
    \begin{equation}
        \mathbb{P}(\{|h - \bar{h}| \leq \theta_{\epsilon}\eta\}\ f.e.) = 1.
    \end{equation} 

    For $\bar{h}(A:C) = \bar{h}(B:C) = \frac{D_A D_B}{2 D_C} [1+ O(D_A^{-2})]$, the proof also has two parts.
    One is $\mathbb{P}[\{h > 0\}\ f.e.] = 1$ for $n_C < 3/4$, another is $\mathbb{P}(h \rightarrow \bar{h}) = 1$.
    Now, we denote $h = \bar{h}(A:C) = \bar{h}(B:C)$.
    For the first, one can follow similar methods to prove it, or notice that $\{h(A:C) = 0\} = \{h(B:C) = 0\} = \{h(A:B) = 0\}$ from the structure theorem of Markov gap $h$, 
    \begin{equation}
        \mathbb{P}[\{h > 0\}\ f.e.] = 1,
    \end{equation} 
    for $h = \bar{h}(A:C) = \bar{h}(B:C)$ follows from $\mathbb{P}[\{h(A:B) > 0\}\ f.e.] = 1$ proved in above.
    For the second one, notice $h(A:C)$ and $h(B:C)$ are $2(2\log D_A + \log D)$-Lipschitz function and $2(2\log D_B + \log D)$-Lipschitz function, respectively
    (recall that we assume $n_A \leq n_B \leq n_C$ at the outset).
    We still denote 
    \begin{equation}
       \theta_{\epsilon} = \frac{\log D}{\sqrt{D^{1-\epsilon}}} \rightarrow 0,
    \end{equation}
    for any $\epsilon > 0$, then it has
    \begin{equation}
        \eta_0 \equiv |\bar{h} - m_h|/\theta_{\epsilon} \leq \frac{(1+2n_X)\sqrt{2 \pi}}{D^{\epsilon/2}} \rightarrow 0,
    \end{equation}
    where $n_X = n_A$ or $n_B$ for $h(A:C)$ or $h(B:C)$, accordingly.
    This similar means that 
    \begin{align}
        & \mathbb{P}[|h - \bar{h}| \geq \theta_{\epsilon}\eta]
        \leq \mathbb{P}[|h - m_h| \geq \theta_{\epsilon}(\eta - \eta_0)] \nonumber\\
        & \leq 2 \exp\left[-\frac{(\eta - \eta_0)^2}{2 (1+2n_X)^2} D^{\epsilon}\right] \rightarrow 0.
    \end{align}
    Similarly, for any $\delta > 0$, there are $N_0$ such that $C_{\delta} = \frac{\eta^2(1- \delta)}{2 (1+2n_X)^2} \leq \frac{(\eta - \eta_0)^2}{2 (1+2n_X)^2}$ for any $N \geq N_0$, $\eta_0 < \eta$.
    \begin{align}
        & \sum_{N=N_0}^{\infty}\mathbb{P}_N(|h - \bar{h}| \geq \theta_{\epsilon}\eta)  \leq \sum_{N=N_0}^{\infty} 2 \exp\left[-C_{\delta} D^{\epsilon}\right] \nonumber \\
        & \leq  \frac{2 (C_{\delta}- \log C_{\delta} +1) e^{-C_{\delta}}}{\epsilon \log 2}  < \infty, 
    \end{align}
    and it follows that
    \begin{equation}
        \mathbb{P}(\{|h - \bar{h}| \leq \theta_{\epsilon}\eta\}\ f.e.) = 1.
    \end{equation} 
\end{proof}
\begin{remark}
    Similar to the remark on the proof of Theorem~\ref{theorem: mutual_information},  the minimal size $N_{\delta,\epsilon}$ of system such that $I$ converges to mean $\bar{I}$ with accuracy $\epsilon$, i.e. $|I - \bar{I}|\leq \epsilon$, and a tolerant probability of failure $\delta$ can be derived from the probability inequalities Eqs.~(\ref{eq: concentration_h_bar}) and~(\ref{eq: concentration_h_theta}).
\end{remark}

\begin{corollary} \label{corollary: undistillable}
    With respect to Haar measure, the state $\ket{\psi_{ABC}}$ with weakly nonzero Markov gap almost surely has an entanglement undistillable state, i.e. bound entangled $\mathrm{BND}_{\mathrm{marg}}$ or separable $\mathrm{SEP}_{\mathrm{marg}}$, as one of its marginal states with only finite exceptions, and the state $\ket{\psi}$ with strongly nonzero Markov gap almost surely has NPT states $\mathrm{NPT}_{\mathrm{marg}}$ as all of its marginal states with only finite exceptions in the thermodynamic limit
    \begin{align}
        & \mathbb{P}[\mathrm{BND}_{\mathrm{marg}} \cup \mathrm{SEP}_{\mathrm{marg}} \vert \{h\rightarrow0^{+}\} \ f.e.] = 1,\\
        & \mathbb{P}[\mathrm{NPT}_{\mathrm{marg}}  \vert \{h\gg 0\} \ f.e.] = 1.
    \end{align}
\end{corollary}\noindent

\begin{proof}
    Let $A$ denote some events on the states in Haar ensemble.
    Then, in the region $n_{\max}<3/4$, the probability of event $A$ is 
    \begin{align}
        \mathbb{P}(A) & = \mathbb{P}(A \cap \{h \not\rightarrow \bar{h}\} \ i.o.)  \nonumber\\
        &~~~~ + \mathbb{P}(A \cap \{h \rightarrow \bar{h}\} \ f.e.) \nonumber\\
        & = \mathbb{P}(A \cap \{h_N > 0\} \ f.e.)
    \end{align}
    with Theorem~\ref{theorem: Markov_gap}.
    For tripartite Haar random state, we denote $C$ as the largest subsystem, $n_C = n_{\max}$, and $A$ as the smallest subsystem, $n_A = n_{\min}$. 
    In particular, since both the thresholds for separability and PPT 
    \begin{equation}
        n_{\mathrm{PPT}} =\frac{1}{2} \leq n_{\mathrm{SEP}} = \frac{1+n_{\min}}{2} <3/4,
    \end{equation}
    the Haar random state has the threshold behaviors for separability and PPT of marginal states
    \begin{equation}
        \left\{\begin{array}{lc}
            \mathbb{P}(\mathrm{NPT}_{\mathrm{marg}}) \rightarrow 1, & n_{C}< n_{\mathrm{PPT}} \\
            \mathbb{P}(\mathrm{BND}_{\mathrm{marg}}) \rightarrow 1, & n_{\mathrm{PPT}} < n_{C} <n_{\mathrm{SEP}}  \\
            \mathbb{P}(\mathrm{SEP}_{\mathrm{marg}}) \rightarrow 1, & n_{C} >n_{\mathrm{SEP}}
        \end{array}\right. 
    \end{equation}
    where $\mathrm{NPT}_{\mathrm{marg}}$ means all the marginal states are negative partial transpose, $\mathrm{BND}_{\mathrm{marg}}$ means one of the marginal states is bound entangled, and $\mathrm{SEP}_{\mathrm{marg}}$ means one of the marginal states is separable.
    In precise, since speed of the convergence of the probabilities is  exponential,
    \begin{equation}
        \mathbb{P}(\cdot) \geq 1 - c\exp[-C 2^{\alpha N}],
    \end{equation} 
    where $\cdot$ denote some events (see following equation), it is not difficult to enhance it with Lemma~\ref{lemma: f.e.}
    \begin{equation}
        \left\{\begin{array}{lc}
            \mathbb{P}(\mathrm{NPT}_{\mathrm{marg}} \ f.e.) = 1, & n_{C}< n_{\mathrm{PPT}} \\
            \mathbb{P}(\mathrm{BND}_{\mathrm{marg}} \ f.e.) = 1, & n_{\mathrm{PPT}} < n_{C} <n_{\mathrm{SEP}}  \\
            \mathbb{P}(\mathrm{SEP}_{\mathrm{marg}} \ f.e.) = 1, & n_{C} >n_{\mathrm{SEP}}
        \end{array}\right. .
    \end{equation}
    Moreover, since 
    \begin{equation}
        \left\{\begin{array}{lc}
            \mathbb{P}(\{h_N>>0\} \ f.e.) = 1, & n_{C}< n_{\mathrm{PPT}} \\
            \mathbb{P}(\{h\rightarrow0^{+}\} \ f.e.) = 1, & n_{\mathrm{PPT}} < n_{C} <3/4  \\
            \mathbb{P}(\{h\rightarrow0\} \ f.e.) = 1, & n_{C} > 3/4
            \end{array}\right. .
    \end{equation} 
    With Eq.~(\ref{eq: event_compose}), it follows that
    \begin{align}
        & \mathbb{P}[\mathrm{BND}_{\mathrm{marg}} \cup \mathrm{SEP}_{\mathrm{marg}} \vert \{h\rightarrow0^{+}\} \ f.e.] = 1,\\
        & \mathbb{P}[\mathrm{NPT}_{\mathrm{marg}}  \vert \{h\gg 0\} \ f.e.] = 1.
    \end{align}
\end{proof}


%

\end{document}